\RequirePackage{fix-cm}
\documentclass{svjour3}
\usepackage{graphicx}

\usepackage{amsthm,amsmath,epsfig,subfigure,amsfonts,amssymb,xcolor,amsbsy,times}
\usepackage{bm}
\usepackage{bbm}
\smartqed
\usepackage{mathrsfs}
\usepackage{comment}
\usepackage{mathtools}
\usepackage{cases}
\usepackage[numbers]{natbib}
\usepackage{hyperref}
\hypersetup{%
	colorlinks=true,
	linkcolor=blue,
	filecolor=blue,
	citecolor=blue,
	breaklinks=true}
\usepackage{endnotes}

\numberwithin{equation}{section}
\numberwithin{theorem}{section}
\numberwithin{lemma}{section}
\numberwithin{proposition}{section}
\numberwithin{definition}{section}
\numberwithin{remark}{section}
\spnewtheorem{assumption}{Assumption}{\bf}{\rm}

\newcommand{\md}{\mathrm{d}}

\newcommand{\mR}{\mathbb{R}}

\newcommand{\mE}{\mathbb{E}}
\newcommand{\mF}{\mathbb{F}}
\newcommand{\mP}{\mathbb{P}}

\newcommand{\mU}{\mathbb{U}}
\renewcommand{\epsilon}{\varepsilon}

\newcommand{\F}{\mathcal{F}}

\newcommand{\B}{\mathcal{B}}
\newcommand{\Ui}{\mathcal{U}}

\newcommand{\cA}{\mathcal{A}}
\newcommand{\cW}{\mathcal{W}}
\newcommand{\cL}{\mathcal{L}}

\newcommand{\cJ}{\mathcal{J}}

\newcommand{\cP}{\mathcal{P}}
\newcommand{\cK}{\mathcal{K}}

\newcommand{\bu}{\mathbf{u}}
\newcommand{\bU}{\mathbf{U}}

\newcommand{\sF}{\mathscr{F}}

\newcommand{\hato}{\hat{\mathrm{o}}}

\newcommand{\mt}{\mathbf{t}}

\newcommand{\bmu}{\bm{\mu}}
\newcommand{\supp}{\mathrm{supp}}
\newcommand{\mVar}{\mathrm{Var}}
\newcommand{\ES}{\mathrm{ES}}
\newcommand{\MES}{\mathrm{MES}}
\usepackage[fontsize=12pt]{fontsize}

\begin{document}

\title{Equilibrium master equations for time-inconsistent problems with distribution dependent rewards
}

\titlerunning{Distribution dependent rewards}        

\author{Zongxia Liang         \and
        Fengyi Yuan 
}


\institute{Zongxia Liang \at
              Department of Mathematical Sciences, Tsinghua University \\
              \email{liangzongxia@mail.tsinghua.edu.cn}           
           \and
           Fengyi Yuan \at
              Department of Mathematical Sciences, Tsinghua University\\
              \email{yfy19@mails.tsinghua.edu.cn}
}

\date{Received: date / Accepted: date}
\maketitle
\begin{abstract}
We provide a unified approach to find equilibrium solutions for time-inconsistent problems with distribution dependent rewards, which are important to the study of behavioral finance and economics. Our approach is based on {\it equilibrium master equation}, a non-local partial differential equation on Wasserstein space. We refine the classical notion of derivatives with respect to distribution and establish It$\hato$'s formula in the sense of such refined derivatives. Our approach is inspired by theories of Mckean-Vlasov stochastic control and mean field games, but is significantly different from both in that: we prohibit marginal distribution of state to be an input of closed loop control; we solve the best reaction to individual selves in an intra-person game instead of the best reaction to large populations as in mean field games. As applications, we reexamine the dynamic portfolio choice problem with rank dependent utility based on the proposed novel approach. We also recover the celebrated extended HJB equation when the reward of the problem has a nonlinear function of expectation while reformulating and weakening the assumptions needed. Most importantly, we provide a procedure to find an equilibrium solution of a dynamic mean-ES portfolio choice problem, which is completely new to the literature.
\end{abstract}

\section{Introduction}
Dynamic optimization in continuous time has been an important mathematical tool for describing many decision problems in economics and finance. Starting from time $t$, wealth $x$ and choosing the strategy $u$, we usually assign a {\it reward} $J(t,x;u)$ to the agent, and the goal for the agent is to maximize $J(t,x;u)$ dynamically. For example, $u$ is the portfolio strategy, and $J(t,x;u)=EU(X^{t,x,u}_T)$ is the expected utility from the resulted terminal wealth. Beyond classical expected utility theory (EUT), there have been enormous alternatives of $J$ in literature, for examples, mean-variance criterion: $J(t,x;u)=EX^{t,x,u}_T-\frac{\gamma}{2}\mathrm{Var}(X^{t,x,u}_T)$ (\citet{Zhou2000}, \citet{Bjork2014}); mean-risk criterion with risk described by other risk measures: $J(t,x;u)=EX^{t,x,u}_T-\gamma \rho(X^{t,x,u}_T)$ (\citet{Jin2005} with $\rho$=nonlinear function of expectation, \citet{He2015} with $\rho=$weighted VaR and \citet{Zhou2017} with $\rho=$VaR); rank dependent utility theory (RDUT) (\citet{Jin2008}, \citet{Hu2021}). All these extensions to EUT have common features: the rewards depend on wealth (outcome) not only through its expectation, but also through moments of higher orders, or more generally, through the distribution of it: $J(t,x;u)=g(\cL(X^{t,x,u}_T))$. Such rewards are called distribution dependent by us, and constitute our main concerns. We propose a master equation approach to study these problems systematically.

Problems with distribution dependent rewards are important because they provide great flexibility to the modeling of risk preference, not limited to expected utility, and this is key to the study of behavioral finance and economics. Problems with distribution dependent rewards are challenging because unless properly reformulated, they are intrinsically {\it time-inconsistent}. Mathematically speaking, if we insist that only $(t,x)$ are treated as state variables, and the strategies are restricted to closed loop controls, the dynamic programming principle is not satisfied anymore. The direct consequences include we can not rely on the celebrated HJB equation, from which value function is solved and the optimal strategy is obtained. Another consequence is the terminology ``dynamic optimality" becomes ill-posed, because the problem at $(0,x_0)$ and that at $(t,x)$ are inconsistent with each other. To overcome this, there are two main approaches in the literature, both are developed very recently:

The first approach is to provide additional information to the agent. More specifically, there is an additional state variable $\mu_t=\cL(X_t)$, the probability law of original state $X_t$. This formulation is called Mckean-Vlasov stochastic control. Closely related to mean field games originated in \citet{Lasry2007}, Mckean-Vlasov control problems have been studied extensively. With state space augmented, time consistency is recovered and dynamic programming equation, a PDE on the Wasserstein space, is obtained, see e.g., Chapter 6 of \citet{Carmona2018} and \citet{Buckdahn2017}. There are also efforts to extend the concept of solution from classical solution to weak solution and viscosity solution (e.g. \citet{Wu2020}, \citet{Mou2019}). One drawback of this approach is that the obtained closed loop optimal strategy contains $\cL(X_t)$ as one argument, which means that to implement such an optimal strategy, the agent needs to know the whole distribution of $X_t$. This does not make sense in the aspect of applications, because typically the agent is only informed of the realization $X_t(\omega)$ (i.e., he observes the wealth he owns), or at most, the whole path of $X_{\cdot\wedge t}$ (i.e., he remembers what he has owned). For a concrete example of this approach, the reader is referred to Subsection 3.3 of \citet{Wu2020}.

The second approach is to properly revise the formulation of the rationality of the agent. Researchers have considered problems where the agent commits himself to the optimal solution at $(0,x_0)$ (see \citet{Jin2008}, \citet{Xu2016} and references therein for the examples of RDUT). More recently, the focus gradually shifts to the so-called equilibrium solution theory, where the agent at each state $(t,x)$ is seen as a player, and the Nash equilibrium of such an {\it intra-person game} is desired. There has been literature in this formulation concerning RDUT (\citet{Hu2021}), quantile maximization (\citet{He2020}),  and nonlinear functions of expectations, including MV criterion as a special case ( \citet{Bjork2017}, \citet{Hernandez2020}, \citet{He2021b}, \citet{He2021}, among others). However, as can be noted from our short review, the studies in this direction are very limited and are conducted case by case. A unification for problems with distribution dependent rewards is needed because distribution dependence is the most important source of time-inconsistency except for the well-studied non-exponential discounting and state dependence\endnote{The readers interested in the review of different sources of time-inconsistency can refer to the introduction part of our recent paper \citet{Liang2021}}. Moreover, such a unification can provide tools for studying other specific problems under intra-person game formulation, some of which remain largely open and unexplored, e.g., portfolio choice with the criterion of the form mean-VaR, mean-ES or mean-$\rho$ for other risk measures $\rho$\endnote{For static mean-$\rho$ portfolio choice, one is referred to \citet{Herdegen2021}}.

In this paper, we establish the desired unification and provide the {\it master equation} describing equilibrium solution (Theorem \ref{EMEthm}), contributing to the second approach mentioned above. Inspired by the master equation approach in mean field games and mean field control, the key idea of our approach is to lift the state space from $[0,T]\times\mR^d$ to $[0,T]\times \cP_2(\mR^d)$, where $\cP_2(\mR^d)$ is the set of probability distributions with finite second order moment. We construct a crucial auxiliary function defined on this lifted space and prove that it satisfies the corresponding master equation. While the idea seems similar to that in Mckean-Vlasov control, our idea has at least three essential differences to theirs: {\it first}, we focus on time consistent equilibrium solution of the intra-person game, where to implement such a strategy, only the realization of $X_t$ is needed, while in Mckean-Vlasov control, the strategy contains $\cL(X_t)$ as an input; {\it second}, in our problem the function $f$ that satisfies the master equation is constructed by us and such a construction is very crucial for our approach (see discussion at the beginning of Section \ref{mainresults} as well as (\ref{audef})), while the function that satisfies master equation (also called HJB equation in literature) in Mckean-Vlasov control is just the value function of the problem and comes out trivially; {\it third}, in this paper the derived master equation is {\it non-local} because $f$ contains information of equilibrium solution $\hat{u}$ and $\hat{u}$ is decided by $f$ through the {\it diagonal} $(t,\delta_x)\in [0,T]\times \cP_2(\mR^d)$, while in Mckean-Vlasov control, master equation is just a PDE on Wasserstein space.

Three examples are provided to argue that our approach is widely applicable. In Subsections \ref{RDUexample} and \ref{nonlinearexpectation}, we recover some results in existing literature. By applying the proposed master equation, Subsection \ref{RDUexample} obtains the same ODE describing so-called {\it risk premium reduction coefficient} $\lambda(\cdot)$ as in \citet{Hu2021}, and Subsection \ref{nonlinearexpectation} confirms the celebrated extended HJB equation in the case where the reward contains a nonlinear function of expectation. In these two examples, either derivation is greatly simplified, or the conditions needed are reformulated and seem to be weaker than the existing work. In Subsection \ref{MES}, we provide a procedure to find the equilibrium solution proportional to wealth of a mean-ES portfolio choice problem, which is completely new to the literature. It turns out that similarly to the RDUT problem, the equilibrium master equation leads to a nonlinear ODE, and the existence of an equilibrium solution proportional to wealth hinges on the existence of a positive solution of such an ODE.

To make our theory applicable to the aforementioned specific problems, we encounter additional difficulties. In one of our applications (see Subsection \ref{RDUexample}), the auxiliary function $f$ may not be defined for general $\mu\in\cP_2(\mR^d)$, but only on a subset $\cP\subset \cP_2(\mR^d)$. We overcome this by proposing weak L-derivatives beyond classical L-derivatives (chapter 5 of \citet{Carmona2018}) and establish It$\hato$'s formula for this notion of weak derivatives (see Subsection \ref{weakLderivative}). The master equation is also stated in the sense of proposed weak L-derivatives (Theorem \ref{EMEthm}). In other examples (Subsections \ref{MES} and \ref{nonlinearexpectation}), such a notion is also very crucial for analyzing, and the underlying set $\cP$ can be designed case by case, providing great flexibility for the studying (see (\ref{MESPdef}) for an example).

There are two related papers that are valuable to be mentioned: \citet{Mei2019} and \citet{Mei2020}. In both papers, authors study time-inconsistent problems with distribution dependent rewards (costs). However, their interpretations of equilibrium solutions are very different from ours. Their solutions are defined as follows:
first fix {\it a prior} $\bmu$, and find the equilibrium strategy $\hat{u}$ and equilibrium trajectory $X_{\cdot}$; then $\hat{u}$ is defined as an equilibrium if $\cL(X_t)=\bmu_t$ (see Definition 4.1 of \citet{Mei2020} or Definition 3.1 of \citet{Mei2019}). It can be seen that their equilibrium essentially interprets the distribution argument as the mean field effect of large populations, hence they essentially study {\it inter-person games} with a large number of players. Therefore, our formulation seems more suitable for problems encountered in mathematical finance, e.g., portfolio choice. In addition, they do not consider intrinsic time-inconsistency brought by distribution dependence itself but focus on non-exponential discounting. Moreover, although they provide nice theoretical results (existence and uniqueness of equilibrium solution), their approach seems extremely hard to be applied to specific problems, while this paper studies concrete examples from mathematical finance, as a supplement to the literature.

Contributions of the present paper can be summarized as follows: {\it first}, we propose notions of weak L-derivatives, weak It$\hato$'s formula, and weak solutions of master equations, which are interesting in their own rights, and are crucial for our approach to study time-inconsistent problems; {\it second}, we provide a unified and novel approach to reexamine many famous time-inconsistent problems that have been studied, including RDUT and MV as examples; {\it third}, we provide powerful tools for problems that are completely new or largely open, such as dynamic mean-$\rho$ portfolio choice problems, or RDUT with heterogeneous distortion and S-shaped utility (Remark \ref{heterodistortion}). In addition, our results can potentially direct the design of a numerical algorithm to find equilibrium solutions that are not explicit or help to establish the theoretical existence of the equilibrium solutions, which is one of the pivotal difficulties in the studying of time-inconsistent problems.

The rest of this paper is organized as follows: in Section \ref{ProblemFormulation} we formulate the problem in a unified way. Section \ref{mainresults} contains the main results of this paper, including the notion of weak L-derivatives and the master equation. In Section \ref{examples}, three concrete problems are studied as examples, which show that our approach is widely applicable. In Section \ref{conclusion}, we conclude this paper. All proofs of results in Section \ref{examples} are presented in appendices, while other proofs are presented in the main text because of their importance.
\vskip 15pt
\section{Problem formulation}\label{ProblemFormulation}
Fixing a time maturity $T>0$, we consider a complete probability space $(\Omega,\mF,\F_T,\mP)$ supporting a $k$-dimensional Brownian motion $W=\{W_t\}_{0\leq t\leq T}$. Consider also a metric space $\bU$ as the {\it space of control}. We assume that the state dynamic under the closed-loop control $u:[0,T]\times \mR^d\to \bU$ is governed by the following stochastic differential equation (SDE):
\begin{equation}\label{dynamics}
	\md X_t = b(t,X_t,u(t,X_t))\md t + \sigma (t,X_t,u(t,X_t))\md W_t,
\end{equation}
where $b:[0,T]\times \mR^d\times \bU\to \mR^d$, $\sigma:[0,T]\times \mR^d\times \bU \to \mR^{d\times k}$ are both measurable. For simplicity we will ease the notation as
\begin{align*}
	& b^u(t,x)=b(t,x,u(t,x)), \\
	&\Sigma^u(t,x)=\sigma^u(t,x)\sigma^u(t,x)^{\mt}=\sigma(t,x,u(t,x))\sigma^{\mt}(t,x,u(t,x)).
\end{align*}
Let $\Ui$ be the set of admissible {\it closed-loop control}, i.e., the allowed strategies that the agent can choose from. For the well-posedness of the problem, we assume $\Ui$ satisfies the following:
\begin{assumption}\label{app1}
	For any $u\in \Ui$, we have:
	\begin{itemize}
		\item[(1)]$\bU\subset \Ui$ if we identify any $\bu\in \bU$ with the constant map in $\Ui$.
		\item[(2)]$b^u$ and $\sigma^u$ are Lipschitz in $x\in\mR^d$, uniformly in t.
		\item[(3)]$b^u$ and $\sigma^u$ are right continuous in $t\in [0,T)$, and are bounded in $t$, uniformly for $x$ in any compact subsets of $\mR^d$.
	\end{itemize}
	For convenience, we denote by $\Ui_0$ the set of all $u$ satisfies Assumption \ref{app1}.
\end{assumption}
\begin{remark}
	The set $\Ui$ will be implicit in most part of the paper. We only use Assumption \ref{app1} to derive the moment estimation of $X$ (see Lemma \ref{momentest}). Except in Subsection \ref{MES}, take $\Ui=\Ui_0$ is enough for our purpose, but the flexibility to {\it design} admissible $\Ui$ provides more possibilities anyway.
\end{remark}
For $u\in \Ui$, we may denote by $X^{t,x,u}$ the strong solution under control $u$ and initial condition $X_t=x$.
Let us define the reward function at state $(t,x)$ with strategy $u$ by $J(t,x;u)\triangleq g(\cL(X_T^{t,x,u}))$. The agent initially faces the following problem:
\[
\max_{u \in \Ui} J(t,x;u),\ \ \forall (t,x)\in [0,T)\times \mR^d.
\]
This problem is in general {\it time-inconsistent}\endnote{See the discussions for special cases: probability distortion in \citet{Hu2021}, mean-variance in \citet{Cui2012}}. The source of time-inconsistency is the distribution argument in the function $g$. In classical control theory, the time-consistency hinges on the tower property of the conditional expectation. Together with the Markovian property of state variables, such a tower property gives the dynamic programming principle satisfied by value function and optimal control. This indicates that, if the reward function $g$ has a distribution argument, we should also include the marginal distribution of $X$ as another augmented state variable. Indeed, stochastic control theory for Mckean-Vlasov SDE confirms our argument (see Chapter 6 of \citet{Carmona2018}). Therefore the dynamic optimal control should have the form $\tilde{u}_t=\tilde{u}(t,X_t,\cL(X_t))$. However, if the agent is not informed of the distribution of $\cL(X_t)$ (which is usually the case in applications), he is not aware of the changes of $\cL(X_t)$, which causes time-inconsistency. To be specific, with initial state $x_0$, he will take the control $\tilde{u}(s,y,\delta_{x_0})$ at {\it any} future state $(s,y)$. This is of course not optimal anymore, although it is indeed optimal at $(0,x_0)$. To resolve the time-inconsistency issue, we choose to utilize the {\it consistent planning} approach, and to find the {\it equilibrium control} $\hat{u}:(t,x)\mapsto \hat{u}(t,x)$, which is time-consistent.
\begin{definition}
	\begin{itemize}
		\item[(1)]For $u\in \Ui$, $\bu\in \bU$ and $\epsilon>0$, we define the {\it spike variation} of $u$ by
		\[
		u_{t,\epsilon,\bu}(s,y)= \bu I_{s\in [t,t+\epsilon)} + u(s,y) I_{s\notin [t,t+\epsilon)}.
		\]
		\item[(2)]A Borel set $G\subset \mR^d$ is said to be $(x_0,\Ui)$-allowed if $\forall u\in \Ui$, $t\in [0,T)$, we have $\nu_t^u(G)=1$, where $\nu^u_t$ is the distribution of $X^{0,x_0,u}_t$.
		\item[(3)]$\hat{u}\in \Ui$ is called an {equilibrium} control with initial wealth $x_0$ if there exists a $(x_0,\Ui)$-allowed set $G$, such that for any $\bu \in \bU$, $(t,x)\in [0,T)\times G$ we have,
		\begin{equation}
			\label{eqcond}
			\limsup_{\epsilon\to 0}\frac{J(t,x;\hat{u}_{t,\epsilon,\bu})-J(t,x;\hat{u})}{\epsilon}\leq 0.
		\end{equation}
	\end{itemize}
\end{definition}
\begin{remark}
	Using our notations defined in Section \ref{mainresults}, we have $\nu^u_t=\phi^{0,u}_t\delta_{x_0}$.
\end{remark}
\begin{remark}
	As revealed by footnotes 6-7 of \citet{He2021}, to require that (\ref{eqcond}) is satisfied for {\it any} $(t,x)\in [0,T]\times \mR^d$ is sometimes too restrictive, because every states, even those that are never accessed by {\it any} admissible trajectory, are considered. They define the reachable set, and only require that (\ref{eqcond}) holds for those $x$ in reachable set of {\it equilibrium} trajectory at any time $t$. Essentially, \citet{Hu2021} also proposes the same requirement. This approach is equivalent to our setting if ``$\nu^u_t$ for any $u\in \Ui$" is revised to ``$\nu^{\hat{u}}_t$ only". However, in a game theoretical perspective, this may leads to an equilibrium that is not subgame perfect. Therefore, we choose to impose the requirement that (\ref{eqcond}) holds for all those states that is {\it possible} to reach by some {\it admissible} (not necessarily equilibrium) control. See also the discussion in Remark 2.7 of \citet{Hernandez2020}, which is essentially the same to our approach.
\end{remark}
\vskip 20pt
\section{Main results}\label{mainresults}
The trick of our approach is to replace the {\it stochastic} dynamic of $X$ with a {\it deterministic} dynamic of its marginal distribution. This is originated from the studying of mean field problems in mathematical physics, and Mckean-Vlasov type problems in stochastic control theory. Surprisingly, it is rather crucial in the studying of distributional-rewarded time-inconsistent problems. We can not treat $\mu$, the marginal distribution of $X$, as a state variable, due to the information constraint mentioned before, but we still need to lift the state space to $\cP_2(\mR^d)$ and treat every state $x\in \mR^d$ as $\delta_x$. Working on this lifted space, we can establish the equation of the auxiliary function, which is a partial differential equation on $[0,T]\times \cP_2(\mR^d)$. This equation is, however, satisfied automatically along any trajectory. Equilibrium condition (\ref{eqcond}) imposes additional requirement on the {\it diagonal} of the augmented space, that is, at $(t,\delta_x)$ for $(t,x)\in [0,T]\times \mR^d$, making the equation non-local.
\subsection{Notations and preliminary lemmas}
Throughout this paper, we equip $\cP_2(\mR^d)$ with the usual 2-Wasserstein metric:
\[
\cW_2(\mu,\nu)=\left\{\inf_{\pi\in\Pi(\mu,\nu)}\left\{\int_{\mR^d\times \mR^d}|x-y|^2\pi(\md x,\md y)\right\}\right\}^{1/2},
\]
where $\mu,\nu\in \cP_2(\mR^d)$, and $\Pi(\mu,\nu)$ is the set of couplings between $\mu$ and $\nu$. That is to say, every $\pi\in \Pi(\mu,\nu)$ is a probability measure on $\mR^d\times \mR^d$ such that for any Borel set $B$, $\pi(B\times \mR^d)=\mu(B)$, $\pi(\mR^d\times B)=\nu(B)$. In addition, we denote by $L^2_{\mu}(\mR^d)$, $H^1_{\mu}(\mR^d)$ (or simply $L^2_{\mu}$, $H^1_{\mu}$ if there is no confusion) the quadratic integrable Lebsgue and Sobolev spaces with respect to the measure $\mu$, respectively. To be specific,
\[
L^2_{\mu}(\mR^d)=\{f:\int_{\mR^d}|f(x)|^2\mu(\md x)<\infty\}=\{f:E|f(\xi)|^2<\infty\},
\]
where $\xi$ is $\mR^d$-valued random variable with distribution $\mu$, and
\[
H^1_{\mu}(\mR^d)=\{f:f\in L^2_{\mu},\partial_xf \in L^2_{\mu}\},
\]
where the gradient is interpreted in the sense of generalized function (see \citet{hormander2015analysis} for mathematical definition).

Based on the well-posedness of (\ref{dynamics}), given initial condition $\cL(X_t)=\mu$ at time $t$, we denote by $\bm{ \mu}^{t,\mu,u}$ the marginal distribution of the unique strong solution of (\ref{dynamics}) with such an initial condition. Therefore, we can find a {\it deterministic} flow map $\phi^{t,u}:\cP_2(\mR^d)\to \cP_2(\mR^d)^{[0,T]}, \mu \mapsto \bm{\mu}^{t,\mu,u}_{\cdot}$. Here, we use a convention that $\bmu^{t,\mu,u}_r=\mu$ for $0\leq r\leq t$. Strong uniqueness (in fact, uniqueness in law is sufficient) implies the flow property: $\phi^{s,u}_r\phi^{t,u}_s=\phi^{t,u}_r$ for any $t\leq s \leq r$. We now establish some trivial yet useful properties of $\phi^{t,u}$. Define the metric $d_T$ on $\cP_2(\mR^d)^{[0,T]}$ by:
\[
d_T(\omega_2,\omega'_2)=\sup_{0\leq s \leq T}\cW_2(\omega_2(s),\omega_2'(s)).
\]
We have the following lemma regarding the continuity of $\phi^{t,u}$.
\begin{lemma}\label{momentest}
	For any $\mu,\nu\in \cP_2(\mR^d)$, $t\in [0,T]$ and $t\leq s_1,s_2\leq T$, there exist constants $C$ (may be varied from line to line) that only depend on $T$ and $u$, such that
	\begin{equation}
		\label{phiest1}
		\cW_2(\phi^{t,u}_{s_1}\mu,\phi^{t,u}_{s_2}\mu)\leq C(1+M^{1/2}_2(\mu))|s_1-s_2|^{1/2}.
	\end{equation}
	\begin{equation}
		\label{phiest2}
		d_T(\phi^{t,u}\mu,\phi^{t,u}\nu)\leq C\cW_2(\mu,\nu).
	\end{equation}
\end{lemma}
\begin{proof}
	The proof is based on the moment estimate of strong solutions (see, e.g., Theorem 1.6.3 from \citet{Yong1999}). For any two random variable $\xi,\xi'\in L^2(\Omega,\F_0)$ with $\cL(\xi)=\mu,\cL(\xi')=\nu$, we denote by $X$ and $X'$ the solution of (\ref{dynamics}) with initial condition $X_t=\xi$, $X'_t=\xi'$, respectively. Under Assumption \ref{app1}, we  have
	\[
	\cW_2(\phi^{t,u}_{s_1}\mu,\phi^{t,u}_{s_2}\mu)^2\leq\mE|X_{s_1}-X_{s_2}|^2\leq C(1+\mE |\xi|^2)|s_1-s_2|.
	\]
	Noting that $\cL(\xi)=\mu$, we have $\mE|\xi|^2=M_2(\mu)$, thus (\ref{phiest1}) is proved. On the other hand, we have
	\[
	d_T(\phi^{t,u}\mu,\phi^{t,u}\nu)^2\leq \sup_{0\leq s\leq T}\mE |X_s-X'_s|^2 \leq C\mE|\xi-\xi'|^2.
	\]
	Taking the infimum on the right hand side with respect to all possible couplings of $\xi$ and $\xi'$, we conclude that (\ref{phiest2}) is true.
\end{proof}

By (\ref{phiest1}), $\phi^{t,u}$ maps $\cP_2(\mR^d)$ into $\Omega_2\triangleq C([0,T];\cP_2(\mR^d))$, the family of continuous mappings from $[0,T]$ to $\cP_2(\mR^d)$. By (\ref{phiest2}), $\phi^{t,u}$ is continuous, hence Borel measurable. Here we equip $\cP_2(\mR^d)$ and $\Omega_2$ with Borel $\sigma$ algebras $\B(\cP_2(\mR^d))$ and $\B(\Omega_2)$, induced by $\cW_2$ and $d_T$ respectively. We define $\mP^{t,\mu,u} \triangleq (\phi^{t,u})_{\#}\delta_{\mu}=\delta_{\phi^{t,u}_{\cdot}\mu}$, which is a probability measure on $(\Omega_2,\B(\Omega_2))$. With a slight abuse of notation, we also denote by $\bmu$ the canonical process on $\Omega_2$. We define $\{\sF_r\}_{0\leq r\leq T}$ to be the filtration generated by $\bmu$. It is clear that $\sF_r\subset \B(\Omega)$ for any $r\in [0,T]$.
\begin{remark}
	In fact there is no probability at all. $\mP^{t,\mu,u}$ in fact concentrates on a single path. However, to ease the notation, and to provide full analogy with the expected utility case, we insist on the probability notation. Another advantage of this notation is that it allows natural extensions to {\it mixed strategies}. Indeed, if strategy is taken as a probability measure $\mU$ on $\Ui$, then we may denote $\mP^{t,\mu,\mU}=\mE^{\mU}\circ \mP^{t,\mu,u}=(\Phi^{t,\mu})_\#\mU$, with $\Phi^{t,\mu}:u\mapsto \phi^{t,u}\mu$. Taking $\mU=\delta_u$, this generalization degenerates to our current setting. This direction is left for future study.
\end{remark}

The following lemma regarding the changes of measure under spike variation will be used to prove our main result.
\begin{lemma}\label{measureep}
	For any $u\in \Ui$, $\bu \in \bU$, $(t,\mu)\in [0,T)\times \cP_2(\mR^d)$ and $\epsilon >0$, we have
	\begin{itemize}
		\item[(1)] $\mP^{t+\epsilon,\mu, u_{t,\epsilon,\bu}} = \mP^{t+\epsilon,\mu,u}$.
		\item[(2)] $\mP^{t,\mu,u_{t,\epsilon,\bu}} = \mP^{t,\mu,\bu}$ on $\sF_{t+\epsilon}$.
	\end{itemize}
\end{lemma}
\begin{proof}
	(1) is from the uniqueness in law of the solution to (\ref{dynamics}). (2) is equivalent to $\phi^{t,u_{t,\epsilon,\bu}}_s=\phi^{t,\bu}_s$, $\forall t\leq s\leq t+\epsilon$. By definition $\phi^{t,u_{t,\epsilon,\bu}}_s \mu= \phi^{t,\bu}_s\mu$ for $t\leq s < t+\epsilon$. By (\ref{phiest1}) this is also true for $s=t+\epsilon$. Thus (2) is proved.
\end{proof}
The following lemma describes the ``Markovian" property of the family of ``probability" measure $\{\mP^{t,\mu,u}\}_{(t,\mu)\in [0,T]\times \cP_2(\mR^d)}$, which will be used later.
\begin{lemma}\label{Markovian}
	For any $(t,\mu,u)\in [0,T]\times \cP_2(\mR^d)\times \Ui$, and any $t\leq s \leq T$, we have
	\[
	\mE^{t,\mu,u}g(\bmu_T)=\mE^{t,\mu,u}\mE^{s,\bmu_s,u}g(\bmu_T).
	\]
\end{lemma}
\begin{proof}
	This proof easily follows from the flow property of $\phi$.
\end{proof}
\begin{remark}
	It may be helpful to rewrite the dynamic of $\bmu$ as a deterministic Fokker-Planck equation. However, all facts we use about $\bmu$, such as It$\hato$'s formula or Lipshitz estimates, are based on that it is the flow of marginal measures of $X$, which is governed by (\ref{dynamics}).
\end{remark}
\subsection{Weak L-derivatives and It${\bf\hato}$'s formula}\label{weakLderivative}
To study distributional-rewarded time-inconsistent problems, we shall introduce the calculus on spaces of probability measures, especially on $\cP_2(\mR^d)$. We first briefly recall the notions of L-derivatives and It$\hato$'s formula along a flow of probability measures, which are well-developed in the literature. This part is mainly adopted from Chapter 5 of the monograph \citet{Carmona2018}, hence we omit all the proofs. However, this is not enough for our purposes, because from practical examples we find sometimes the function we study is not L-differentiable in the classical sense (see Section \ref{examples} for details). Therefore we introduce the notion of weak L-derivatives, and prove that It$\hato$'s formula still holds under this notion of differentiability.
\begin{definition}\label{lderivativeDEF}
	Consider a function $u:\cP_2(\mR^d)\to \mR$ and $\mu_0\in \cP_2(\mR^d)$. If there exists a lifting $\tilde{u}:L^2(\tilde{\Omega},\tilde{\F},\tilde{\mP})\to \mR$, i.e., $\tilde{u}(X) = u(\cL(X))$, such that $\tilde{u}$ is Fr$\acute{\rm e}$chet differentiable near some $X_0$ with $\cL(X_0)=\mu_0$, and the mapping $D \tilde{u}:X\mapsto D \tilde{u}(X)$ is continuous in the sense of $L^2$ at $X_0$, then $u$ is said to be L-continuously differentiable at $\mu_0$. We denote by $\partial_{\mu}u(\mu_0)\triangleq D\tilde{u}(X_0)$ its L-derivative at $\mu_0$.
\end{definition}

\begin{proposition}\label{Lderivativefunction}
	The L-continuously differentiability is well-defined, i.e., independent of the choice of lifting. Moreover, if $u$ is L-continuously differentiable at $\mu_0$, then there exists a $\xi \in L^2_{\mu_0}(\mR^d)$ such that $D\tilde{u}(X_0) = \xi(X_0)$ for any $X_0$ with $\cL(X_0)=\mu_0$. Thus, the L-derivative of $u$ at $\mu_0$ can be identified with an $L^2_{\mu_0}(\mR^d)$ function $\xi: \mR^d \to \mR^d$.
\end{proposition}
Now we will state the It$\hato$'s formula along a flow of probability measures, under the assumptions that $f$ satisfies sufficient regularity.
\begin{definition}
	Fix $t_0,t_1\in [0,T)$. Consider a function $f:[t_0,t_1] \times \cP_2(\mR^d)\to \mR$. $f$ is said to be in the class $C^{1,1}([t_0,t_1];H^1)$ if:
	\begin{itemize}
		\item[(1)]For any $\mu\in \cP_2(\mR^d)$, the function $t\mapsto f(t,\mu)$ is of class $C^1$, with $\partial_tf$ being jointly continuous in $(t,\mu)$.
		\item[(2)]For any $t\in [t_0,t_1]$, the function $\mu \mapsto f(t,\mu)$ is L-continuously differentiable, and, for any $\mu\in \cP_2(\mR^d)$, we can find a version of the function $v\mapsto \partial_{\mu}f(t,\mu)(v)$ such that $(t,\mu,v)\mapsto \partial_{\mu}f(t,\mu)(v)$ is locally bounded and continuous at any $(t,\mu,v)$ with $v\in \supp(\mu)$.
		\item[(3)]For the version of $\partial_{\mu}f$ mentioned above and for any $(t,\mu)\in[t_0,t_1]\times \cP_2(\mR^d)$, the function $v\mapsto \partial_{\mu}f(t,\mu)(v)$ is continuously differentiable and $(t,\mu,v)\mapsto \partial_v\partial_{\mu}f(t,\mu)(v)$ is continuous at $(t,\mu,v)$ with $v\in \supp(\mu)$. Moreover, for any bounded $\cK\subset  \cP_2(\mR^d)$ we have
		\begin{equation}\label{Lboundedness}
			\sup_{(t,\mu)\in [t_0,t_1]\times \cK}\|   \partial_{\mu}f(t,\mu) \|_{H^1_{\mu}(\mR^d)}<\infty.
		\end{equation}
	\end{itemize}
\end{definition}
\begin{remark}
	If there is no confusion, we simplify $C^{1,1}([t_0,t_1];H^1)$ to $C^{1,1}(H^1)$.
\end{remark}
\begin{proposition}[It$\hato$'s formula]
	\label{Ito}
	Fix $0\leq t<t_0<T$, $\mu \in \cP_2(\mR^d)$, $u\in \Ui$. For any $f\in C^{1,1}([t,t_0];H^1)$,  we have $\mP^{t,\mu,u}$-almost surely,
	\begin{eqnarray}\label{L-Ito}
		f(s,\bmu_s) =f(t,\mu)+\int_t^s (\partial_t+\cL^u)f(r,\bmu_r)\md r,\ \ \forall s\in [t,t_0],
	\end{eqnarray}
	Here, we use the notation:
	\begin{equation}\label{Ludef}
		\cL^{u}f(t,\mu)=\int_{\mR^d}\left(     \partial_{\mu}f(t,\mu)(y)\cdot b^{u}(t,y)+\frac{1}{2}\mathrm{tr}(\partial_y\partial_{\mu}f(t,\mu)(y)\Sigma^{u}(t,y))\right)\mu(\md y).
	\end{equation}
\end{proposition}
One drawback of the notion of L-derivatives is that we require $f$ is well-defined for any $\mu\in \cP_2(\mR^d)$. In applications (see Section \ref{examples} for details) we need to evaluate $\partial_{\mu}f $ for auxiliary function $f$. However, it is very likely that $f$ is only defined on some specific subset $\cP\subset \cP_2(\mR^d)$ rather than on the whole $\cP_2(\mR^d)$. To resolve this issue, we propose the notion of weak L-derivatives and develop It$\hato$'s formula here.
\begin{definition}
	\label{wlderivative}
	For $\cP\subset \cP_2(\mR^d)$, a function $f:[t_0,t_1] \times \cP \to \mR$ is said to have weak L-derivative on $[t_0,t_1]$ for $0\leq t_0<t_1<T$, denoted by $f\in W^{1,1}([t_0,t_1]\times \cP;H^1)$, if $f(\cdot,\mu)\in C^0\cap H^1$ for each $\mu\in \cP$ and there exist a sequence $f^N\in C^{1,1}([t_0,t_1];H^1)$ such that,
	\begin{itemize}
		\item[(1)] $f^N \to f$, locally uniformly for $(t,\mu)$. That is, for any bounded $\cK\subset  \cP$,
		\begin{equation}\label{uniformuniform}
			\sup_{(t,\mu)\in [t_0,t_1]\times \cK}|f^N(t,\mu)-f(t,\mu)|\to 0,\ \mbox{as}\ N\to \infty.
		\end{equation}
		\item[(2)]  $\partial_tf^N\to \partial_tf$, uniformly in $\mu$ for bounded subsets of $\cP$, and $L^2$ in $t$. That is, for any bounded $\cK\subset  \cP$,
		\begin{equation}\label{uniformL2}
			\sup_{\mu\in   \cK}\int_{t_0}^{t_1}|\partial_t f^N(t,\mu)-\partial_tf(t,\mu)|^2 \md t\to 0, \ \mbox{as}\ N\to \infty.
		\end{equation}
		\item[(3)]For any bounded $ \cK\subset \cP$, there exists a constant $C=C_{\cK,t_0}$ such that
		\begin{equation}\label{uniformbounded1}
			\sup_{(t,\mu)\in [t_0,t_1]\times \cK,N\geq 1}\|\partial_{\mu}f^N(t,\mu)\|_{H^1_{\mu}}<\infty.
		\end{equation}
		\item[(4)]For any fixed $(t,\mu)\in [t_0,t_1]\times  \cP$, $\partial_{\mu}f^N(t,\mu)\to g$ in $H^1_{\mu}(\mR^d)$ norm. We denote $\partial_{\mu}f(t,\mu)=g$.
	\end{itemize}
\end{definition}
\begin{remark}
	The weak L-derivative of $f\in W^{1,1}([t_0,t_1]\times \cP;H^1)$ is well-defined, i.e., $g$ does not depend on the sequence of $f^N$ as long as $f^N\to f$ uniformly on bounded sets. To see this, let us fix $(t,\mu)\in [t_0,t_1]\times \cP_2(\mR^d)$ and pick $f^N\to f$, $h^N \to f$ uniformly on bounded sets, and denote by $\tilde{f}^N$, $\tilde{h}^N$ their lifting to $L^2$ respectively. By choosing $\xi\sim \mu$, $\xi'\in L^2$, $\|\xi'\|_{L^2}\leq 1$ and $\xi'\to 0$ in $L^2$, we have
	\begin{align*}
		&\tilde{f}^N(t,\xi+\xi')-\tilde{f}^N(t,\xi)=\tilde{\mE}\partial_{\mu}f^N(t,\mu)(\xi)\xi'+o(\|\xi'\|_{L^2}),\\ \\
		&\tilde{h}^N(t,\xi+\xi')-\tilde{h}^N(t,\xi)=\tilde{\mE}\partial_{\mu}h^N(t,\mu)(\xi)\xi'+o(\|\xi'\|_{L^2}).
	\end{align*}
	Defining $R_N= 2\sup_{\nu:\cW_2(\mu,\nu)\leq 1}(|f^N(t,\nu)-f(t,\nu)|+|h^N(t,\nu)-f(t,\nu)|)$, we have
	\[
	|\tilde{\mE}(\partial_{\mu}h^N(t,\mu)(\xi)-\partial_{\mu}f^N(t,\mu)(\xi))\xi'| \leq o(\|\xi'\|_{L^2})+R_N.
	\]
	Taking supremum with respect to all $\xi'\in L^2$ such that $\|\xi'\|\leq \sqrt{R_N}$, we have
	\[
	\|\partial_{\mu}h^N(t,\mu)-\partial_{\mu}f^N(t,\mu)\|_{L^2_{\mu}}\leq o(1)+\sqrt{R_N},\ \mbox{as } \ N\to \infty.
	\]
	By uniform convergence, $R_N\to 0$ as $N\to\infty$. Therefore, if $\partial_{\mu}f^N\to g_1$, $\partial_{\mu}h^N\to g_2$ in $H^1_{\mu}$, we have
	\[
	\|g_1-g_2\|_{L^2_{\mu}}\leq \liminf_{N\to \infty}\|\partial_{\mu}h^N(t,\mu)-\partial_{\mu}f^N(t,\mu)\|_{L^2_{\mu}}=0.
	\]
\end{remark}
The requirement (2) in Definition \ref{wlderivative} is very restrictive and hard to verify. However, due to the following lemma, convergence of derivatives can be omitted if we know a prior the convergence rate of $f^N\to f$, which is expected to be verified case by case in applications.
\begin{lemma}\label{weakLdlemma}
	Suppose that $f(\cdot,\mu)\in C^0\cap H^1$, and for any bounded $ \cK\subset \cP$, we have
	\begin{equation}\label{uniformcontinuity}
		\sup_{\substack{(s,t)\in [t_0,t_1], \mu\in \cK\\ |s-t|<\epsilon}} |f(s,\mu)-f(t,\mu)|\to 0,\epsilon \to 0.
	\end{equation}
	\begin{equation}\label{uniformboundL2}
		\sup_{\mu\in \cK}\|\partial_t f(\cdot,\mu)\|_{L^2(t_0,t_1)}<\infty
	\end{equation}
	If there exists a sequence $f^N\in C^{1,1}(H^1)$ such that for any $\cK\subset  \cP$ bounded, there exist constants $C=C_{\cK}$ and $\delta=\delta_{\cK}>0$ such that
	\begin{equation}\label{decaycondition}
		\sup_{(t,\mu)\in [t_0,t_1]\times \cK}|f^N(t,\mu)-f(t,\mu)|\leq C/N^{1+\delta},
	\end{equation}
	\begin{equation}\label{uniformboundedness}
		\sup_{(t,\mu)\in [t_0,t_1]\times \cK,N\geq 1 }\|   \partial_{\mu}f^N(t,\mu) \|_{H^1_{\mu}(\mR^d)}<\infty.
	\end{equation}
	and for any $(t,\mu)\in [t_0,t_1]\times \cP$, $\partial_{\mu}f^N(t,\mu)\to g$ in $H^1_{\mu}(\mR^d)$, then $f\in W^{1,1}([t_0,t_1]\times \cP;H^1)$.
\end{lemma}
\begin{proof}
	Let us consider the temporal modifier $\eta$ defined by
	\[
	\eta(t)=\left\{
	\begin{array}{ll}
		\frac{1}{C}\exp(\frac{1}{t^2-1}), &t^2<1, \\
		0,&t^2\geq 1,
	\end{array}
	\right.
	\]
	where $C$ is a normalization constant such that $\int_{\mR}\eta = 1$. Now consider $\eta^N(t)=N\eta(Nt)$. Clearly we have $\eta^N\in C^{\infty}_c(\mR)$. Moreover, we have,
	\[\int_{\mR}|\partial_t \eta^N| = N^2\int_{\mR}|\partial_t\eta^1(Nt)|\md t=CN.
	\]
	In the rest of this proof, we denote by $h_\mu$ the function $h(\cdot,\mu)$ for any $h:[t_0,t_1]\times  \cP\to \mR$, and consider their zero extensions to $\mR$ if needed. Now consider $\tilde{f}^N(t,\mu) = f^N_{\mu}*\eta^N(t)$. We first prove that for any fix $(t,\mu)$ and $N$, $\partial_{\mu}\tilde{f}^N(t,\mu)=\partial_{\mu}f^N_{\mu}*\eta^N(t)$. For convenience we still denote by $\tilde{f}^N$ and $f^N$ their lifting to $L^2$. For $\xi,\xi'\in L^2$, $\xi\sim \mu$ and $\|\xi'\|_{L^2}\leq 1$ we have,
	\begin{align*}
		&\left| \tilde{f}^N(t,\xi+\xi')-\tilde{f}^N(t,\xi)-\tilde{\mE}\left[\int_{\mR}
		Df^N(t-s,\xi)\eta^N(s)\md s\right]\xi'\right| \\
		=& \left|\int_{\mR}\int_0^1\tilde{\mE}Df^N(t-s,\xi+\lambda \xi')\xi' \eta^N(s)\md \lambda\md s- \tilde{\mE}\left[\int_{\mR}
		Df^N(t-s,\xi)\eta^N(s)\md s\right]\xi'           \right| \\
		\leq& \int_{\mR}\int_0^1 \tilde{\mE}|Df^N(t-s,\xi+\lambda \xi')-Df^N(t-s,\xi)||\xi'|\eta^N \md \lambda\md s  \\
		\leq& \|\xi'\|_{L^2}\int_{\mR}\int_0^1\left[\tilde{\mE}|D f^N(t-s,\xi+\lambda\xi')-Df^N(t-s,\xi)|^2\right]^{1/2}\eta^N(s)\md \lambda\md s.
	\end{align*}
	Let $R(s,t,\lambda,\xi')=\left[\tilde{\mE}|D f^N(t-s,\xi+\lambda\xi')-Df^N(t-s,\xi)|^2\right]^{1/2}$.  Because $f^N$ is known to be L continuously differentiable, for fix $(s,t,\lambda)$, $R(s,t,\lambda,\xi')\to 0$ as $\|\xi'\|_{L^2}\to 0$. By (\ref{Lboundedness}),
	\[
	\sup_{0\leq \lambda\leq 1}R(s,t,\lambda,\xi')\leq 2\sup_{\substack{s:|s-t|^2<1/N^2\\
			\nu:\cW_2(\nu,\mu)\leq 1}} 2\|\partial_{\mu}f^N(s,\nu)\|_{L^2_{\nu}}<\infty.
	\]
	By Dominated convergence theorem,
	\[
	\int_{\mR}\int_0^1R(s,t,\lambda,\xi')\eta^N(s)\md s \to 0, \|\xi'\|_{L^2}\to 0,
	\]
	\[
	\left|\tilde{f}^N(t,\xi+\xi')-\tilde{f}^N(t,\xi)-\tilde{\mE}\left[\int_{\mR}
	Df^N(t-s,\xi)\eta^N(s)\md s\right]\xi'\right|\leq \|\xi'\| o(1)=o(\|\xi'\|_{L^2}),
	\]
	as $\|\xi'\|_{L^2}\to 0$. This proves the assert $\partial_{\mu}\tilde{f}^N(t,\mu)=\partial_{\mu}f^N_{\mu}*\eta^N(t)$. By the property of $\eta^N$, it is straightforward to verify by definition that $\tilde{f}^N\in C^{1,1}(H^1)$. Using the relation $\partial_{\mu}\tilde{f}^N(t,\mu)=\partial_{\mu}f^N_{\mu}*\eta^N(t)$ and Dominated convergence theorem, we can prove that $\partial_v\partial_{\mu}\tilde{f}^N(t,\mu)=(\partial_v\partial_{\mu}f^N_{\mu})*\eta^N(t)$. Using Jensen's inequality and the fact: $\int_{\mR}\eta^N=1$ for any $N$, we have
	\begin{align*}
		\| \partial_{\mu}\tilde{f}^N(t,\mu)-\partial_{\mu}\tilde{f}^{M} (t,\mu) \|^2_{H^1_{\mu}}\leq &C\left[ \int_{\mR}\|\partial_{\mu}f^N(t-s,\mu)-\partial_{\mu}f^M(t-s,\mu)\|^2_{H^1_{\mu}}\eta^N(s)\md s\right.\\
		&+\left.\int_{\mR}\|\partial_{\mu}f^M(t-s,\mu)\|^2_{H^1_{\mu}}|\eta^N(s)-\eta^M(s)|\md s \right].
	\end{align*}
	Combining (\ref{uniformboundedness}) and Dominated convergence theorem, we conclude that $\partial_{\mu}\tilde{f}^N(t,\mu)$ converges to some $\tilde{g}$ in $H^1_{\mu}$ norm, for each fixed $(t,\mu)\in [t_0,t_1]\times \cP$. The only thing remains to be proved is the convergence of $\tilde{f}^N$ and their derivatives. Indeed, for $ \cK\subset \cP$ bounded and for any $(t,\mu)\in [0,T]\times \cK$,
	\[
	|f^N_{\mu}*\eta^N(t)-f_{\mu}(t)|\leq |f^N_{\mu}*\eta^N-f_{\mu}*\eta^N|+|f_{\mu}*\eta^N-f_{\mu}|\leq I_1^N(t)+I_2^N(t),
	\]
	where
	\begin{align*}
		\sup_{(t,\mu)\in [t_0,t_1]\times \cK}I_1^N(t)&\leq \sup_{(t,\mu)\in [t_0,t_1]\times \cK} \int_{\mR}|f^N(t-s,\mu)-f(t-s,\mu)|\eta^N(s)\md s \\
		&\leq
		\sup_{(t',\mu)\in[t_0,t_1]\times \cK}|f^N(t',\mu)-f(t',\mu)|\\
		&\leq C/N^{1+\delta}\to 0, \ \mbox{as}\ N \to \infty,\\
		\sup_{(t,\mu)\in [t_0,t_1]\times \cK}I^N_2(t)&\leq \int_{\mR}|f(s,\mu)-f(t,\mu)|\eta^N(t-s)\md s,\\
		&\leq \sup_{\substack{s:|s-t|^2\leq 1/N^2\\ (t,\mu)\in [t_0,t_1]\times \cK} } |f(s,\mu)-f(t,\mu)|\to 0, \ \mbox{as}\ N \to \infty,
	\end{align*}
	here we have used (\ref{uniformcontinuity}) and (\ref{decaycondition}). Therefore
	\[
	\sup_{(t,\mu)\in [t_0,t_1]\times \cK}|\tilde{f}^N(t,\mu)-f(t,\mu)|\to 0, \ \mbox{as}\ N \to \infty .
	\]
	On the other hand, noting that $\partial_t (f^N_{\mu}*\eta^N)=f^N_{\mu}*\partial_t \eta^N$, we  have
	\begin{eqnarray*}
	|\partial_t(f^N_{\mu}*\eta^N)(t)-\partial_tf_{\mu}(t)|&\leq & |f^N_{\mu}*\partial_t\eta^N-f_{\mu}*\partial_t\eta^N|+|\partial_tf_{\mu}*\eta^N-\partial_tf_{\mu}|\\
&\triangleq & I_3^N(t,\mu)+I_4^N(t,\mu)
	\end{eqnarray*}
	with
	\begin{align*}
		\sup_{(t,\mu)\in [t_0,t_1]\times \cK}I_3^N(t,\mu)&\leq \sup_{(t,\mu)\in [t_0,t_1]\times \cK} \int_{\mR}|f^N(t-s,\mu)-f(t-s,\mu)|\partial_t\eta^N(s)\md s \\
		&\leq
		CN\sup_{(t',\mu)\in[t_0,t_1]\times \cK}|f^N(t',\mu)-f(t',\mu)|\\
		&\leq C_{\cK}/N^{\delta}\to 0, \ \mbox{as}\ N \to \infty.
	\end{align*}
	Denote by $\sF$ the Fourier transform, then after proper zero extensions,
	\begin{align*}
		\sup_{\mu\in\cK}\|I_4\|_{L^2(t_0,t_1)}
		&=\sup_{\mu\in \cK}\|\partial_t f_{\mu}*\eta^N-\partial_tf_{\mu}\|_{L^2(\mR)}\\&=\sup_{\mu\in \cK}\|\sF(\partial_t f_{\mu})(\sF(\eta^N)-1)\|_{L^2(\mR)} \\
		&\leq  \|\sF(\eta^N)-1\|_{L^{\infty}(\mR)}\sup_{\mu\in \cK}\|\partial_tf_{\mu}\|_{L^2(t_0,t_1)}.
	\end{align*}
	For any given $t'\in \mR$, we have
	\begin{align*}
		|\sF(\eta^N)(t')-1|&\leq C\int_{\mR}|e^{itt'}-1|\eta^N(t)\md t \\
		&\leq C\int_{\mR}|e^{itt'/N}-1|\eta(t)\md t\to 0,\ \mbox{as}\ N \to \infty,
	\end{align*}
	where Dominated convergence theorem has been used. Then (\ref{uniformL2}) is proved. Thus, the proof follows.
\end{proof}
Now we are ready to state and prove the { weak} It$\hato$'s formula, which is crucial for the establishment of equilibrium master equation in Subsection \ref{emesection}, and for applications in Section \ref{examples}. To begin with, we need the concept of {\it invariant set}.
\begin{definition}
	For $u\in \Ui$, $0\leq t<t_0<T$, a subset $\cP\subset \cP_2(\mR^d)$ is said to be $([t,t_0],u)$-invariant if for any $ t\leq s \leq t_0$, $\phi^{t,u}_s\cP \subset \cP$.
\end{definition}
\begin{remark}
	$\cP_2(\mR^d)$ itself is clearly $([0,T],u)$-invariant for any $u\in \Ui$. In fact by usual moment estimate, $\cP_p(\mR^d)$ for any $p\geq 2$ are also $([0,T],u)$-invariant. Nontrivial examples of invariant set (of distributions) under stochastic flow are interested in their own right and remain largely unexplored to the best of our knowledge. In the present paper, we will choose very specific locally invariant sets (see Definition \ref{localinvariant}) in the examples (see Subsections \ref{RDUexample} and \ref{MES} for details).
\end{remark}
\begin{theorem}[Weak It$\hato$'s formula]
	\label{weakIto}
	Using the same notations as in Proposition \ref{Ito}, for any $([t,t_0];u)$-invariant $\cP$, $f\in W^{1,1}([t,t_0]\times \cP;H^1)$, we have $\mP^{t,\mu,u}$-almost surely,
	\begin{equation}\label{ItoM}
		f(s,\bmu_s) = f(t,\mu)+\int_t^{s} (\partial_t+\cL^u)f(r,\bmu_r)\md r,
	\end{equation}
	$\forall s\in [t,t_0]$, where the operator $\cL^u  $ is defined by (\ref{Ludef}).
\end{theorem}
\begin{proof}
	By definition of $W^{1,1}([t,t_0]\times\cP;H^1)$, we  choose a sequence $f^N\in C^{1,1}(H^1)$ such that $f^N\to f$ locally uniformly, and $\partial_t f^N \to \partial_t f$ locally uniformly in $\mu$, $L^2$ in $t$. Moreover, $\partial_{\mu} f^N(t,\mu)\to \partial_{\mu}f(t,\mu)$ in $H^1_{\mu}$ for each $(s,\mu)\in [t,t_0]\times \cP$. By Proposition \ref{Ito}, $\forall s\in [t,t_0]$,
	\begin{equation}\label{ItoN}
		f^N(s,\bmu_{s}) =f^N(t,\mu)+\int_t^{s} (\partial_t+\cL^u)f^N(r,\bmu_r)\md r.
	\end{equation}
	Clearly, $f^N(s,\bmu_{s})\to f(s,\bmu_{s})$ and $f^N(t,\mu)\to f(t,\mu)$. To handle the last term in (\ref{ItoN}), we note that
	\[
	\left| \int_t^{s}(\partial_t+\cL^u)f^N(r,\bmu_r)\md r-\int_t^{s}(\partial_t+\cL^u)f(r,\bmu_r)\md r \right| \leq J^N_1+J^N_2,
	\]
	where
	\begin{align*}
		J^N_1&\leq \int_t^{s}(|\partial_tf^N(r,\bmu_r)-\partial_t f(r,\bmu_r)|)\md r \\
		&\leq \sqrt{T} \left[\sup_{\nu:\cW_2(\nu,\mu)\leq C\sqrt{T} }\left\{\int_t^s|\partial_t f^N(t,\nu)-\partial_tf(t,\nu)|^2\md t\right\}\right]^{1/2}\to 0,\  \mbox{as}\ N\to \infty,\\
		J^N_2&\leq  \int_t^{s}\tilde{\mE}|\partial_{\mu}f^N(r,\bmu_r)(\tilde{X}_r)-\partial_{\mu}f(r,\bmu_r)(\tilde{X}_r)||b^u(r,\tilde{X}_r)| \md r \\
		&+\frac{1}{2}\int_t^{s}\tilde{\mE}|\partial_v
		\partial_{\mu}f^N(r,\bmu_r)(\tilde{X}_r)-\partial_v
		\partial_{\mu}f(r,\bmu_r)(\tilde{X}_r)||\sigma^u(r,\tilde{X}_r)|^2 \md r,
	\end{align*}
	where $\tilde{X}_r\sim \bmu_r$ under $\tilde{\mE}$. By definition, $\partial_{\mu}f^N(r,\bmu_r)\to \partial_{\mu}f(r,\bmu_r)$ in $H^1_{\bmu_r}$-norm, furthermore,
	\begin{align*}
		&\tilde{\mE}|\partial_{\mu}f^N(r,\bmu_r)(\tilde{X}_r)
		-\partial_{\mu}f(r,\bmu_r)(\tilde{X}_r)||b^u(r,\tilde{X}_r)| \to 0, \  \mbox{as}\ N\to \infty, \\
		&\tilde{\mE}|\partial_v\partial_{\mu}f^N(r,\bmu_r)(\tilde{X}_r)
		-\partial_v\partial_{\mu}f(r,\bmu_r)(\tilde{X}_r)||\sigma^u(r,\tilde{X}_r)|^2 \to 0, \  \mbox{as}\ N\to \infty.
	\end{align*}
	Using boundedness condition (\ref{uniformbounded1}), the polynomial intergrability of $b^u$ and $\sigma^u$ (which is clear from Lipschitz condition and usual moment estimates of SDE) and Dominated convergence theorem, we conclude  $J^2_N\to 0$  as $ N\to \infty$. Thus, we complete the proof.
\end{proof}
\begin{remark}\label{boundedremark}
	The boundedness of $\cK$ in Definition \ref{wlderivative} can be replaced by the boundedness in the following sense: for any $p\geq 2$, there exists $\mu_0\in \cK$ such that $\sup_{\mu\in \cK}\cW_p(\mu,\mu_0)<\infty$. Here $\cW_p$ is the usual $p$-Wasserstein metric defined similarly as $\cW_2$. This is because in the proof of Theorem \ref{weakIto}, we only require that convergence is uniform with respect to the measure family $\{\bmu_r\}_{t\leq t\leq t_0 }$, which clearly satisfies the above boundedness. This remark is only used in Subsection \ref{nonlinearexpectation}.
\end{remark}
\subsection{Equilibrium master equations}\label{emesection}
Given a candidate strategy $\hat{u}$, we define the {\bf auxiliary function} by
\begin{equation}\label{audef}
	f(t,\mu)=\mE^{t,\mu,\hat{u}}g(\bmu_T).
\end{equation}
We will see from the proof of Theorem \ref{EMEthm} that, auxiliary function $f$ is related to {\it value} $J(t,x;\hat{u})$ by
\[
J(t,x;\hat{u})=f(t,\delta_x),
\]
and more importantly,
\[
J(t,x;\hat{u}_{t,\epsilon,\bu})=\mE^{t,\delta_x,\bu}f(t+\epsilon,\bmu_{t+\epsilon}).
\]
It is seen that the construction of auxiliary function is necessary and crucial for master equation approach.

Formally, {\bf equilibrium master equation} is proposed as follows:
\begin{align}
	&(\partial_t +\cL^{\hat{u}})f(t,\mu)=0,\label{EME1}\\
	&\sup_{\bu\in \bU}(\partial_t +\cL^{\bu})f(t,\delta_x)= 0,\label{EME2}
\end{align}
where the operator $\cL^u$ is given in (\ref{Ludef}). We now state and prove the main results of this subsection.
\begin{definition}\label{localinvariant}
	A subset $\cP\subset \cP_2(\mR^d)$ is said to be $\hat{u}\wedge \bU$-locally invariant at $t\in [0,T)$ if for any $\bu \in \bU$, there exists $\epsilon$ such that $\cP$ is both $([t,t+\epsilon];\hat{u})$-invariant and $([t,t+\epsilon];\bu)$-invariant.
\end{definition}
\begin{theorem}\label{EMEthm}
	Fix $\hat{u}\in \Ui$, suppose that we can find a $(x_0,\Ui)$-allowed set $G$ such that for any $0\leq t<T$ there exists  a family of $\{\cP_j(t)\}_{j\in \cJ}$ which are all $\hat{u}\wedge \bU$-locally invariant at $t$ and satisfies $\{\delta_x\}_{x\in G}\subset \cP(t)\triangleq \bigcup_{j\in \cJ}\cP_{j}(t)\subset \cP_2(\mR^d)$, and the corresponding auxiliary function $f\in W^{1,1}([t,t+\epsilon]\times\cP_j(t);H^1)$ for any $j\in J$ and some sufficiently small $\epsilon$ (allowed to be dependent on $j$). If $\hat{u}$ is an equilibrium with initial wealth $x_0$, then:
	\begin{itemize}
		\item[(1)] $f$ satisfies (\ref{EME1}) for any $(t,\mu)\in \{t\in[0,T),\mu\in\cP(t)\}$.
		\item[(2)] $f$ satisfies (\ref{EME2}) for any $t\in [0,T)$ and $x\in G$.
	\end{itemize}
	Conversely, if the associated auxiliary function $f$ satisfies (2) above, then $\hat{u}$ is an equilibrium with initial wealth $x_0$.
\end{theorem}
\begin{proof}
	{\bf \underline{Step1.}} We first prove that for any fixed $\hat{u}\in \Ui$, (\ref{EME1}) holds. For $\mu\in \cP_j(t)$, $0<\epsilon<T-t$, applying Lemma \ref{Markovian} with $s=t+\epsilon$, we have
	\begin{equation}
		\label{Markovian1}
		f(x,\mu)=\mE^{t,\mu,\hat{u}}\mE^{t+\epsilon,\bmu_{t+\epsilon},\hat{u}}g(\bmu_T)=\mE^{t,\mu,\hat{u}}f(t+\epsilon,\bmu_{t+\epsilon}).
	\end{equation}
	Based on (\ref{Markovian1}), applying weak It$\hato$'s formula and  then taking expectation $\mE^{t,\mu,\hat{u}}$ on both hand sides of (\ref{ItoM}) (which is allowed by the fact that $\cP_j(t)$ is $\hat{u}$ invariant locally at $t$), we have
	\begin{equation}\label{epsilonM}
		0=\mE^{t,\mu,\hat{u}}\left[f(t+\epsilon,\bmu_{t+\epsilon})-f(t,\mu)\right]=\mE^{t,\mu,\hat{u}}\int_t^{t+\epsilon}(\partial_t+\cL^{\hat{u}}) f(r,\bmu_r)\md r.
	\end{equation}
	Dividing by $\epsilon$ on both hand sides of Eq.(\ref{epsilonM}), and letting $\epsilon\to 0$ and using Dominated convergence theorem, it is clear that
	(\ref{EME1}) is true for any $t\in [0,T)$, $\mu\in \cP_j(t)$. Therefore it is also true for any $(t,\mu)\in \{t\in [0,T),\mu\in \cP(t)\}$.
	
	{\bf \underline{Step2.}} We now calculate the limit on the left hand side of (\ref{eqcond}). Clearly, $\cL(X^{t,x,\hat{u}}_T)=\phi^{t,\hat{u}}_T\delta_x$. Thus by definition, $J(t,x;\hat{u})=\mE^{t,\delta_x,\hat{u}}g(X_T,\bmu_T)=f(t,\delta_x)$. On the other hand,
	\begin{align*}
		J(t,x;\hat{u}_{t,\epsilon,\bu})&=\mE^{t,\delta_x,\hat{u}_{t,\epsilon,\bu}}g(\bmu_T)\\
		&=\mE^{t,\delta_x,\hat{u}_{t,\epsilon,\bu}}\mE^{t+\epsilon,\bmu_{t+\epsilon},\hat{u}_{t,\epsilon,\bu}}g(\bmu_T)\\
		&=\mE^{t,\delta_x,\hat{u}_{t,\epsilon,\bu}}\mE^{t+\epsilon,\bmu_{t+\epsilon},\hat{u}}g(\bmu_T)\\
		&=\mE^{t,\delta_x,\hat{u}_{t,\epsilon,\bu}}f(t+\epsilon,\bmu_{t+\epsilon})\\
		&=\mE^{t,\delta_x,\bu}f(t+\epsilon,\bmu_{t+\epsilon}),
	\end{align*}
	where we have used Lemma \ref{measureep} and Lemma \ref{Markovian}. Applying weak It$\hato$'s formula (\ref{ItoM}), we get
	\[
	\lim_{\epsilon\to 0}\frac{J(t,x;\hat{u}_{t,\epsilon,\bu})-J(t,x;\hat{u})}{\epsilon}=\lim_{\epsilon\to 0}\frac{\mE^{t,\delta_x,\bu}f(t+\epsilon,\bmu_{t+\epsilon})-f(t,\delta_x)}{\epsilon}=(\partial_t+\cL^{\bu}) f(t,\delta_x).
	\]
	
	{\bf \underline{Step 3.}} We now get back to the main theorem. If $\hat{u}$ is an equilibrium, then by definition and results in Step 2, for any $\bu\in \bU$, $t\in [0,T)$ and $x\in G$, there exists $j\in \cJ$ such that $\delta_x\in \cP_j(t)$. From $f\in W^{1,1}([t,t+\epsilon]\times\cP_j(t);H^1)$ we know $(\partial_t+\cL^{\bu})f(t,\delta_x)\leq 0$. However by Step 1 we know $(\partial_t+\cL^{\hat{u}})f(t,\delta_x)=0$. Thus (\ref{EME2}) holds for any $t$ and $x\in G$. Conversely, if there exists $\hat{u}\in \Ui$ such that the corresponding $f$ given by (\ref{audef}) satisfies (\ref{EME2}), by Step 2 we have, for any $\bu\in \bU$, $t\in [0,T)$  and $x\in G$,
	\[
	\lim_{\epsilon\to 0}\frac{J(t,x;\hat{u}_{t,\epsilon,\bu})-J(t,x;\hat{u})}{\epsilon}\leq 0.
	\]
	Thus equilibrium condition is verified, which completes the proof.
\end{proof}
\begin{remark}
	When applying Theorem \ref{EMEthm} to find equilibrium, the derivatives of $f$ with respect to distribution argument are always evaluating at $\{\delta_x\}_{x\in G}$, and they are usually well-defined there. Thus the lengthy and technical arguments (those about invariant set, for example) that we have just established are purely theoretical, though they are crucial for the introduction of Theorem \ref{EMEthm}.
\end{remark}
\begin{remark}\label{numericalthoughts}
	Consider the operator $\Theta_1:\hat{u}\mapsto f$, where $f$ is given by (\ref{audef}). Consider also the operator $\Theta_2:f\mapsto u$, where $u(t,x)=\mathrm{argmax}_{\bu\in \bU}\{(\partial_t+\cL^{\bu})f(t,\delta_x) \}$. Then the equilibrium is described by the fixed point of $\Theta_2\circ \Theta_1$. This could be the start point for the proof of existence of equilibrium solutions under the current setting. On the other hand, this also gives an algorithm to compute the equilibrium numerically: for any initial $u$, apply $\Theta_1$ to evaluate the corresponding reward $f$ (known as policy evaluation in machine learning); apply $\Theta_2$ to the resulted $f$ (possibly with optimization algorithm) to get another $\tilde{u}$ (known as policy improvement for time-consistent optimization problems). Implementing enough steps of iteration, it is hoped that $u$ converges to the equilibrium. The main challenges are to compute the L-derivatives efficiently. This direction is left for future study.
\end{remark}
\begin{remark}
	The equilibrium master equations also admit natural extension to include other sources of time-inconsistency. For example, one can easily include non-exponential discounting and state dependency by assuming $g=g(s,y,\mu)$ and considering $J(t,x;u)=\mE^{t,\mu,u}g(t,x,\bmu_T)$. Another popular source of time-inconsistency, {\it nonlinear functions of expectation} which we will discuss in Subsection \ref{nonlinearexpectation}, is trivially embedded into our formulation by choosing $g(\mu)=F(E_{\xi\sim \mu}\tilde{g}(\xi))$ for some nonlinear function $F$.
\end{remark}
\vskip 15pt
\section{Examples}\label{examples}
\subsection{Rank dependent utility theory}\label{RDUexample}
In this section we apply our theoretical results to a dynamic portfolio problem under rank dependent utility theory (RDUT), which has been studied in \citet{Hu2021}. We find that equations describing the risk-premium reduction coefficient $\lambda(\cdot)$ comes naturally from our equilibrium master equation (\ref{EME2}). Thus our work can be seen as a (general) theoretical foundation to the case study in \citet{Hu2021}. We suppose that the dynamic of $X$, interpreted as the wealth, follows
\begin{equation}\label{wealth}
	\md X_t=\mu(t)\theta(t,X_t)\md t+\sigma(t)\theta(t,X_t)\md W_t,
\end{equation}
where $\mu$ and $\sigma$ are deterministic functions satisfying appropriate intergablility conditions. That is, the agent invests in a market with interest rate $r=0$, return $\mu(\cdot)$ and volatility $\sigma(\cdot)$. We consider the rank dependent utility
\[
g(\mu)=\int_0^{\infty}w(P_{\xi\sim \mu}(\xi\geq z))u'(z)\md z+ \int_{-\infty}^0[w(P_{\xi\sim \mu}(\xi\geq z))-1]u'(z)\md z,
\]
where $u(x)=1-\frac{1}{\alpha}e^{-\alpha x}$ is the exponential utility function, and $w$ the distortion function, satisfying Assumption \ref{assumptionw} below. Here and afterwards, we use the convention that when operating on functions of $\xi$, $P_{\xi\sim \mu}$ and $E_{\xi\sim \mu}$ represents any probability and expectation operators under the distribution $\xi\sim \mu$. $P_{\xi\sim \mu,\eta\sim \nu}$ and $E_{\xi \sim \mu,\eta\sim \nu}$ are defined similarly, with additional assumption that $\xi$ and $\eta$ are independent unless noted otherwise.
\begin{assumption}\label{assumptionw}
	$w$ is increasing, continuous differentiable with $w(0)=0$, $w(1)=1$, and there exist constants $C>0$ and $\beta\in (0,1)$ such that
	\begin{equation}\label{growthw}
		w'(p)\leq C\left(\frac{1}{p^{1-\beta}}+\frac{1}{(1-p)^{1-\beta}}\right),\forall p\in (0,1).
	\end{equation}
\end{assumption}
\begin{remark}\label{heterodistortion}
	In the most general RDUT, we can consider different distortion $w_+$ and $w_-$ and different utility $u_+$ and $u_-$ for gain and loss probability respectively. In the present paper we assume $w_-=w_+$, $u_+=u_-$ for simplicity. It is very interesting to study the heterogeneous distortion with $S$ shaped utility case and try to find dynamic equilibrium strategy, which is left for future research.
\end{remark}
Consider a special strategy $\theta(t,x)\equiv\theta(t)$, i.e., the investment which is independent from the current wealth. Under this particular $\theta$,  the auxiliary function is
\begin{equation}\label{fdefRDUT}
	f(t,\mu)=\int_0^{\infty}w(p_{\mu}(t,z))u'(z)\md z+\int_{-\infty}^{0}[w(p_{\mu}(t,z))-1]u'(z)\md z,
\end{equation}
where
\begin{align*}
	&p_{\mu}(t,z)=P_{\xi\sim \mu,\eta  \sim N(0,1)}\left(\xi+\Gamma(t)+\sqrt{\Lambda(t)}\eta\geq z\right ),\\
	&\Gamma(t)=\int_t^T \mu(s)\theta(s)\md s,\\
	&\Lambda(t)=\int_t^T |\sigma(s)\theta(s)|^2 \md s.
\end{align*}
\begin{remark}
	Note that $f(t,\mu)$ is not even well-defined for general $\mu\in \cP_2(\mR)$. This is actually our motivation to introduce the notion of weak L-derivatives. Examples include heavy tail distributions with $P_{\xi\sim \mu}(\xi\leq -z)\sim z^{-3}$ as $z\to \infty$, $\xi\leq 0$. In this case
	\[
	E_{\xi \sim \mu}\xi^2=-\int_{-\infty}^0z P_{\xi\sim \mu}(\xi\leq z)\md z\sim \int_{-\infty}^0\frac{|z|}{1+|z|^3}\md z<\infty.
	\]
	However, if we take $1-w(p)\sim p^{\beta}$ as $p\to 1$, then
	\[
	\int_{-\infty}^0\left[w(P_{\xi\sim \mu}(\xi\geq z))-1\right]u'(z)\md z\sim\int_{-\infty}^0 \frac{e^{\alpha |z|}}{(1+|z|^3)^\beta}\md z=\infty.
	\]
	Thus $f(T,\mu)$ is not defined.
\end{remark}
For this specific problem, define
\begin{equation}
	\cP_G(t)=\left \{\mu \mathrm{\ is\ Gaussian,\ with\ Var}(\mu)<\frac{\beta\Lambda(t)}{2(1-\beta)}\right\}.
\end{equation}
\begin{lemma}
	For any $0\leq t_1<T$, $\cP_G(t_1)$ is $\theta'$-invariant locally at any $t\in [0,T)$, for any wealth independent strategy $\theta'$.
\end{lemma}
\begin{proof}
	If the strategy is independent from wealth, $\xi\sim \mu\in \cP_G(t_1)$, we have
	\[
	X^{t,\xi,\theta'}_r=\xi+\Theta^{\theta'}(r)+\sqrt{\Sigma^{\theta'}(r)}\eta
	\]
	with  continuous $\Theta^{\theta'}(\cdot)$ and $\Sigma^{\theta'}(\cdot)$ satisfying $\Theta^{\theta'}(t)=\Sigma^{\theta'}(t)=0$ and $\eta\sim N(0,1)$, independent of $\xi$. As an example, if $\theta'=\theta$ we have $\Theta^{\theta'}(\cdot)=\Gamma(t)-\Gamma(\cdot)$, $\Sigma^{\theta'}(\cdot)=\Lambda(t)-\Lambda(\cdot)$. Therefore we conclude Var$(X^{t,\xi,\theta'}_r)=\mathrm{Var}(\mu)+\Sigma^{\theta'}(r)$. Because $\mathrm{Var}(\mu)<\frac{\beta\Lambda(t)}{1-\beta} $, it is clear that this will also hold for distribution of $X^{t,\xi,\theta'}$ once we choose $r$ sufficiently close to $t$.
\end{proof}
The following result leads us to the equilibrium master equations.
\begin{proposition}\label{RDUTLderivative}
	With $f$ defined in (\ref{fdefRDUT}), we have for any $0\leq t_0<t_1<T$, $\epsilon_0>0$ sufficiently small, $f\in W^{1,1}([t_0,t_1-\epsilon_0]\times \cP_G(t_1);H^1)$ and for any $\mu \in \cP_G(t_1)$,
	\begin{align}
		&\partial_{\mu}f(t,\mu)(y)=E_{\eta\sim N(0,1)}w'\left(p_{\mu}(t,\sqrt{\Lambda(t)}\eta+\Gamma(t)+y)\right )e^{-\alpha(\sqrt{\Lambda(t)}\eta+\Gamma(t)+y)},\\
		&\partial_v\partial_{\mu}f(t,\mu)(y)=E_{\eta\sim N(0,1)}w'\left( p_{\mu}(t,\sqrt{\Lambda(t)}\eta+\Gamma(t)+y)\right )\eta e^{-\alpha(\sqrt{\Lambda(t)}\eta+\Gamma(t)+y)}
	\end{align}
\end{proposition}
\begin{proof}
	See Appendix \ref{proofRDUT}.
\end{proof}
As a direct consequence, when $\mu=\delta_x$, we have
\begin{align*}
	&p_{\delta_x}(t,z)=1-\Phi\left(\frac{z-x-\Gamma(t)}{\Lambda(t)}\right),\\
	&\partial_{\mu}(t,\delta_x)(x)=E_{\eta\sim N(0,1)}w'\left(1-\Phi(\eta)\right)e^{-\alpha\left(\sqrt{\Lambda(t)}\eta+\Gamma(t)+x\right )},\\
	&\partial_v\partial_{\mu}(t,\delta_x)(x)=\frac{1}{\sqrt{\Lambda(t)}}E_{\eta\sim N(0,1)}w'\left(1-\Phi(\eta)\right)\eta e^{-\alpha\left(\sqrt{\Lambda(t)}\eta+\Gamma(t)+x\right )}.
\end{align*}
Here and afterwards, $\Phi$ is the cumulative distribution function of standard normal distribution. We see from (\ref{EME1}) and (\ref{EME2})  that
\[
\theta(t)=\mathrm{argmax}_{\theta\in \mR}\left\{\theta \mu(t)\partial_{\mu}(t,\delta_x)(x)+\frac{1}{2}\theta^2\sigma(t)^2\partial_v
\partial_{\mu}(t,\delta_x)(x)\right\},
\]
and one necessary condition is
\begin{equation}\label{thetadef}
	\theta(t)=\frac{-\mu(t)\partial_{\mu}f(t,\delta_x)(x)}{\sigma(t)^2\partial_v\partial_{\mu}f(t,x,\delta_x)(x)}=\frac{ \mu(t)}{\alpha\sigma(t)^2}\cdot \lambda(t),
\end{equation}
where
\begin{align*}
	&\lambda(t)=\alpha^2\sqrt{\Lambda(t)}\frac{h(\sqrt{\Lambda(t)})}{h'(\sqrt{\Lambda(t)})},\\
	&h(x)=E_{\eta\sim N(0,1)}w'(\Phi(\eta))e^{\alpha \eta x},
\end{align*}
and $\Lambda$ satisfies the following equation:
\begin{equation}\label{lambdaODE}
	\left\{
	\begin{aligned}
		&\Lambda'(t)=-\alpha^2\left[\frac{\mu(t)}{\sigma(t)}\cdot \frac{h(\sqrt{\Lambda(t)})}{h'(\sqrt{\Lambda(t)})}\right]^2\Lambda(t),\\
		&\Lambda(T)=0.
	\end{aligned}
	\right.
\end{equation}
To make $\theta(\cdot)$ defined in (\ref{thetadef}) an equilibrium, one also needs $\partial_v\partial_{\mu}(t,\delta_x)(x)< 0$, which is equivalent to $h'(\sqrt{\Lambda(t)})\geq 0$. This is assured under Assumption 4.2 in \citet{Hu2021}, which is satisfied by many well-known distortion function. We omit the details here.
\vskip 10pt
\subsection{Dynamic mean-ES portfolio choice}\label{MES}
In this subsection we present another example that is closely related to probability distortion, which we call dynamic mean-ES (MES, for short) portfolio choice. Here, ES represents {expected shortfall}, which is now a popular risk measure. To be specific, for $\mu\in \cP_2(\mR)$ and $\alpha_0\in (0,1)$, we define
\[
\ES_{\alpha_0}(\mu)=-\frac{1}{\alpha_0}\int _0^{\alpha_0}F_{\mu}^{-1}(\alpha)\md \alpha,
\]
where $F_{\mu}^{-1}$ is the right continuous quantile function of $\mu$ \endnote{See e.g. footnote 5 of \citet{Xu2016} for specific definition. Here we choose to define expected shortfall following \citet{He2015}. There are also other definitions of expected shortfall, e.g., in \citet{Zhou2017}. These differences are not essential.}. The objective function is given by
\[
g(\mu)= E_{\xi\sim \mu}\xi-\gamma \ES_{\alpha_0}(\mu),
\]
where $\gamma>0$ is the risk aversion, similar to the classical MV setting. That is to say, the agent has a trade off between the expected return and the risk, which is represented by expected shortfall. Portfolio choice under such a criterion is extremely important in both aspects of academic and industrial application. The dynamic version of this problem, where time-inconsistency occurs and the intra-person equilibrium is desired, is new to the literature. Our approach turns out to be powerful for the studying of this problem. To appropriately state the portfolio choice problem, we fix a $\underline{x}\leq 0$, which is the lowest wealth level allowed. We consider the {\it proportion of excess earnings} invested into the risky asset as the strategy. In other words, we consider the following dynamic instead of (\ref{wealth}):
\begin{equation}\label{wealthp}\tag{\ref{wealth}'}
	\md X_t = \mu(t)\theta(t,X_t)(X_t-\underline{x})\md t +\sigma(t)\theta(t,X_t)(X_t-\underline{x})\md W_t.
\end{equation}
Denoting by $X^{t,x,\theta}$ the solution of (\ref{wealthp}), we choose the admissible set $\Ui$ as
\[
\Ui=\left\{\theta\in \Ui_0:X^{t,x,\theta}_s>\underline{x} \mathrm{\ a.\ s.\ },\forall s\in [t,T],x>\underline{x}      \right \}.
\]
If we further define $G=(\underline{x},\infty)$, it is clear that $G$ is a $(\underline{x},\Ui)$-allowed set. On the other hand, the underlying set of probability measure is chosen as
\begin{equation}\label{MESPdef}
	\cP_{\MES}=\bigcup_{\delta>0}\cP_{\MES}^{\delta},
\end{equation}
where
\[
\cP_{\MES}^{\delta}=\left\{\mu\in \cP_2(\mR):E_{\xi\sim \mu}|\xi-\underline{x}|^{-2}<\delta^{-2}\right\}.
\]
For the candidate equilibrium strategy $\hat{\theta}(t,x)=\theta(t)$ (independent of wealth), we list the following notations dependent on $\hat{\theta}$, while we choose not to write such a dependence explicitly:
\begin{align*}
	&X^{t,\xi}_T=\underline{x}+(\xi-\underline{x})e^{\Gamma_1(t)+\sqrt{\Lambda_1(t)}\eta},\ \eta\sim N(0,1)\ \  \mbox{\ being  independent\ of \ } \xi,\\
	&\Gamma_1(t)=\int_t^T\big [\mu(s)\theta(s)-\frac{1}{2}\theta(s)^2\sigma(s)^2\big ]\md s,\\
	&\Lambda_1(t)=\int_t^T\theta(s)^2\sigma(s)^2 \md s.
\end{align*}
It is clear that $\cP_{\MES}^{\delta}$ is locally invariant under any wealth independent strategy $\theta$. Therefore, by Theorem \ref{EMEthm}, (\ref{EME2}) can be used to determine $\theta(t)$. We consider the auxiliary function :
\begin{equation}\label{MESfdef}
	f(t,\xi)=EX^{t,\xi}_T+\frac{\gamma}{\alpha_0}\int_0^{\alpha_0}F^{-1}_{X^{t,\xi}_T}(\alpha) \md \alpha.
\end{equation}
It is seen that $f$ is law-invariant, hence can be seen as a functional with distribution argument. Therefore we can equivalently write $f=f(t,\mu)$. By direct calculation, we have
\begin{equation}\label{FXTdef}
	F_{X^{t,\xi}_T}(z)=E\left[1-H(t,\xi,z)\right],
\end{equation}
with
\begin{equation}\label{Hdef}
	H(t,y,z)=\left\{
	\begin{aligned}
		&\left[1-\Phi\left(\frac{\log(z-\underline{x})-\log(y-\underline{x})-\Gamma_1(t)}{\sqrt{\Lambda_1(t)}}\right)      \right]I_{\{y>\underline{x}\}}, &z>\underline{x},\\
		&I_{\{y\leq \underline{x}\}}, &z=\underline{x},\\
		&I_{\{y\geq \underline{x}\}}+\Phi\left(\frac{\log(\underline{x}-z)-\log(\underline{x}-y)-\Gamma_1(t)}{\sqrt{\Lambda_1(t)}}\right) I_{\{y<\underline{x}\}},&z<\underline{x}.
	\end{aligned}
	\right.
\end{equation}
The following proposition leads us to the highly nonlinear ODE determining $\theta(t)$:
\begin{proposition}\label{MESprop}
	With $f$ defined in (\ref{MESfdef}), for any  $\delta>0$ and sufficiently small $\epsilon>0$, we have $f\in W^{1,1}([t,t+\epsilon]\times\cP^{\delta}_{\MES};H^1)$, and for any $\xi\sim \mu\in \cP_{\MES}$, we have
	\begin{align}
		&\partial_{\mu}f(t,\mu)(y)=Ee^{\Gamma_1(t)+\sqrt{\Lambda(t)}\eta}+\frac{\gamma}{\alpha_0}\int_{-\infty}^{F^{-1}_{X^{t,\xi}_T}(\alpha_0)}\partial_y H(t,y,z)\md z,\label{MESpartialmuf}\\
		&\partial_v\partial_{\mu}f(t,\mu)(y)=\frac{\gamma}{\alpha_0}\int_{-\infty}^{F^{-1}_{X^{t,\xi}_T}(\alpha_0)}\partial_{yy} H(t,y,z)\md z. \label{MESpartialvpartialmuf}
	\end{align}
	In particular, if we take $\mu=\delta_x$ for $x>\underline{x}$, then
	\begin{align}
		&\partial_{\mu}f(t,\delta_x)(x)=Ee^{\Gamma_1(t)+\sqrt{\Lambda(t)}\eta}+\frac{\gamma}{\alpha_0}\int_{-\infty}^{\Phi^{-1}(\alpha_0)}e^{\sqrt{\Lambda_1(t)}z+\Gamma_1(t)}\Phi'(z)\md z,\label{MESpartialmufx}\\
		&\partial_v\partial_{\mu}f(t,\delta_x)(x)=\frac{\gamma}{\alpha_0(x-\underline{x})}\int_{-\infty}^{\Phi^{-1}(\alpha_0)}\left[\frac{z}{\sqrt{\Lambda_1(t)}}-1\right]e^{\sqrt{\Lambda_1(t)}z+\Gamma_1(t)}\Phi'(z)\md z.\label{MESpartialvpartialmufx}
	\end{align}
\end{proposition}
\begin{proof}
	See Appendix \ref{MESproof}.
\end{proof}
To find equilibrium strategy, for $x>0$, we define
\begin{align*}
	&F_1(x)=E_{\eta\sim N(0,1)} e^{\eta x}=e^{x^2/2},\\
	&F_2(x)=\int_{-\infty}^{\Phi^{-1}(\alpha_0)}e^{xz}\Phi'(z)\md z=e^{x^2/2}\Phi(\Phi^{-1}(\alpha_0)-x),\\
	&F_3(x)=\int_{-\infty}^{\Phi^{-1}(\alpha_0)}\left[\frac{z}{x}-1\right]e^{xz}\Phi'(z)\md z=-\frac{1}{x}e^{x^2/2}\Phi'\left(\Phi^{-1}(\alpha_0)-x \right).
\end{align*}
Then, by (\ref{EME2}), $\theta$ should satisfy:
\begin{equation}\label{MESsolution}
	\theta(t)=-\frac{\mu(t)(x-\underline{x})}{\sigma(t)^2(x-\underline{x})^2}\frac{\partial_{\mu}f(t,\delta_x)(x)}{\partial_v\partial_{\mu}f(t,\delta_x)(x)}=-\frac{\mu(t)}{\sigma(t)^2}\frac{F_1(\sqrt{\Lambda_1(t)})+\frac{\gamma}{\alpha_0}F_2(\sqrt{\Lambda_1(t)})}{\frac{\gamma}{\alpha_0}F_3(\sqrt{\Lambda_1(t)})}.
\end{equation}
Rephrasing into an equation of $\Lambda_1(t)$, and using definition of $\Lambda_1(t)$, we have
\begin{equation}\label{MESODE}
	\Lambda_1'(t)=-\Lambda_1(t)\left[\frac{\mu(t)}{\sigma(t)}\cdot \frac{1+\frac{\gamma}{\alpha_0}\Phi(\Phi^{-1}(\alpha_0)-\sqrt{\Lambda_1(t)})}{\frac{\gamma}{\alpha_0}\Phi'(\Phi^{-1}(\alpha_0)-\sqrt{\Lambda_1(t)})}  \right]^2
\end{equation}
with terminal condition $\Lambda_1(T)=0$. By Theorem \ref{EMEthm}, if (\ref{MESODE}) admits a positive solution $\Lambda_1(t)$, then the dynamic MES problem has an equilibrium solution given by (\ref{MESsolution}). Otherwise, dynamic MES problem does not have equilibrium strategy in the form of (\ref{MESsolution}). In fact, it is highly likely that (\ref{MESODE}) does not have positive solution (the zero solution being unique) because unlike in (\ref{lambdaODE}), the right hand is non-singular here, i.e., the denominator is bounded below from 0. We choose to leave the details for future studying.
\begin{remark}
	Following the same procedure, but appropriately revising the formulation of problem, it can be verified that intrinsically, dynamic MES problem does not admit equilibrium strategy that investing a fixed {\it amount} of money into risky asset.
\end{remark}
\subsection{Nonlinear functions of expectations}\label{nonlinearexpectation}
In this subsection we consider a specific form of reward function with general controlled dynamics. Let
\[
g(\mu)=E_{\xi\sim \mu}g(\xi)+ F(E_{\xi\sim \mu}\xi),
\]
where $g\in C^0(\mR)$, $F\in C^1(\mR)$, and the dynamic of $X$ be like (\ref{dynamics}). With the convention in Section \ref{ProblemFormulation}, we have the auxiliary function $f$ taking the form
\begin{equation}\label{nonlinearfdef}
	f(t,\mu)=E_{\xi\sim \mu}g(X^{t,\xi}_T)+F(E_{\xi\sim \mu}X^{t,\xi}_T).
\end{equation}
In \citet{Bjork2017} and \citet{He2021} as well as many other seminal papers on time-inconsistent control problems with nonlinear functions of expectations, the so-called extended HJB equations are derived, which incorporate the nonlinear function $F$ in the equations. One of the key points is to introduce the functions
\begin{equation}\label{h12def}
	h_1(t,x)=\mE^{t,x,\hat{u}}g(X_T)=g(X^{t,x,\hat{u}}_T),\  h_2(t,x)=\mE^{t,x,\hat{u}}X_T=\mE X^{t,x,\hat{u}}_T.
\end{equation}
In this section, we will show that our equilibrium master equation degenerates to extended HJB under some weaker conditions. First of all, we shall introduce the set of measures we are interested in. Let
\[
\cP_0=\left \{\mu\in \cP_2(\mR^d):\int_{\mR^d}|x|^q\mu(\md x)<\infty, \forall q\geq 2   \right \}=\bigcap_{q\geq 2}\cP_q(\mR^d).
\]
From standard SDE theory (with Lipshitz coefficient, e.g., in \citet{Yong1999}), we know $\cP_0$ is $([0,T'];u)$-invariant for any $u\in \Ui$ and $T'>0$. Now we state the assumptions we impose on $h_1$ and $h_2$, which is weaker than classical assumptions because it requires only integrability with respect to measures in $\cP_0$. Moreover, it is more flexible if the desired invariant set $\cP_0$ can be further reduced in some specific problems.
\begin{definition}
	For a function $f:[0,T]\times \mR^d\to \mR, (t,x)\mapsto f(t,x)$, we say that $f\in H^{1,2}_{\cP_0,\infty}$ if:
	\begin{itemize}
		\item[(1)] 	$f\in C^{0,2}([0,T]\times \mR^d)$, whose derivative with respect to $t$ is defined everywhere on $[0,T]$.
		\item[(2)] For any $t\in [0,T]$, $\mu\in \cP_0$, $f(t,\cdot),\partial_t f(t,\cdot),\partial_x f(t,\cdot),\partial_{xx}f(t,\cdot)\in L^2_{\mu}$;
		\item[(3)] For any bounded $\cK\subset \cP_0$ (in the sense of Remark \ref{boundedremark}), $\sup\limits_{\mu\in \cK,t\in [0,T]}\|f(t,\cdot)\|_{H^{1,2}_{\mu}}<\infty$, where
		\[
		\|f(t,\cdot)\|^2_{H^{1,2}_{\mu}}\triangleq \|f(t,\cdot)\|^2_{L^2_{\mu}}+ \|\partial_tf(t,\cdot)\|^2_{L^2_{\mu}}+ \|\partial_xf(t,\cdot)\|^2_{L^2_{\mu}}+ \|\partial_{xx}f(t,\cdot)\|^2_{L^2_{\mu}}.
		\]
	\end{itemize}
\end{definition}
With $f$ defined in (\ref{nonlinearfdef}) and $h_1$, $h_2$ given in (\ref{h12def}), we have the following result, which leads us directly to the corresponded equilibrium master equation.
\begin{proposition}\label{propnonlinear}
	Suppose $h_i\in H^{1,2}_{\cP_0,\infty}$ \ for $i=1,2$, we have $f\in W^{1,1}([0,T]\times \cP_0;H^1)$, and
	\begin{align}
		&\partial_{\mu}f(t,\mu)(y)=\partial_x h_1(t,y)+F'\left (E_{\xi\sim \mu}h_2(t,\xi)\right )\partial_x h_2(t,y),\label{Lderivativenonlinear1}\\
		&\partial_v\partial_{\mu}f(t,\mu)(y)=\partial_{xx} h_1(t,y)+F'\left (E_{\xi\sim \mu}h_2(t,\xi)\right )\partial_{xx} h_2(t,y)\label{Lderivativenonlinear2}
	\end{align}
\end{proposition}
\begin{proof}
	See Appendix \ref{proofnonlinear}.
\end{proof}
In particular, if we take $\mu=\delta_x$ for $x\in \mR^d$, then (\ref{EME2}) becomes
\[
\sup_{\bu\in \bU}\left\{ \cA^{\bu}_Xh_1(t,x)+F'(h_2(t,x))\cA^{\bu}_Xh_2(t,x)   \right \}=0,
\]
which is seen to be exactly the extended HJB in \citet{He2021} (see (3.7) and (3.9) therein). Here $\cA^u_X$ is the generator associated with the diffusion $X$, i.e.,
\[
\cA^u_X = \partial_t + b^u(t,x)\partial_x+\frac{1}{2}\mathrm{Tr}\left(\Sigma^u(t,x)\partial_{xx}\right ).
\]
\begin{remark}
	It is not hard to see that if $f\in C^{1,2}$ and its derivatives up to order 2 in space, and up to order 1 in time, are all of polynomial growth in $x$, then $f\in  H^{1,2}_{\cP_0,\infty}$. Therefore, we slightly reduce the assumptions needed with the help of weak It$\hato$'s calculus established in the present paper.
\end{remark}
\vskip 15pt
\section{Conclusion}\label{conclusion}
Problems with distribution dependent rewards appear naturally and widely and are intrinsically time-inconsistent. This paper provides a unified and novel approach to the study of these problems. Our approach is based on equilibrium master equations, from which equilibrium strategies are obtained. Moreover, we propose the notion of weak L-derivatives and establish weak It$\hato$'s formula. For applications, we use our approach to reexamine two problems that have been studied and investigate a problem that is completely new to the literature.

The present paper also initiates many possible directions for future studies. One direction is to use the equilibrium master equation to study other problems with distribution dependent rewards, which contains greatly many problems from mathematical finance or economics and can be far more complicated and challenging than problems in Section \ref{examples} of this paper. Another direction is to establish the existence of equilibrium solutions and design numerical algorithms to find them based on our thoughts in Remark \ref{numericalthoughts}. This is important because as can be seen from existing studies including this paper, unless equilibrium solutions are solved explicitly, existence is either not proved or based on restrictive assumptions which are not true in specific problems. The last direction we want to mention is to improve theoretical generality, e.g., to permit non-Markovian settings or path dependence. Clearly, this direction is very challenging and needs more advanced mathematical tools such as path dependent It$\hato$'s formula, Mckean-Vlasov BSDEs, etc.
\vskip 25pt
\theendnotes
\vskip 10pt
\bibliographystyle{plainnat}
\bibliography{reference}
\appendix
\renewcommand{\theequation}{\thesection.\arabic{equation}}
\section{Proof of Proposition \ref{RDUTLderivative}}
\label{proofRDUT}
In this section we provide a proof of Proposition \ref{RDUTLderivative}, which is crucial to the solution of RDUT problem under exponential utility. For preparation, we need several technical lemmas.
\begin{lemma}\label{estnormcdf}
	There exists a universal constant C($=1/\sqrt{2\pi}$) such that for any $x>0$,
	\[
	\frac{Cx}{1+x^2}e^{-x^2/2}\leq 1-\Phi(x)\leq \frac{C}{x}e^{-x^2/2}.
	\]
\end{lemma}
\begin{proof}
	This is a well-known result. For completeness we provide a short proof as follows:\\
	Let $f_1(x)=\frac{x}{1+x^2}e^{-x^2/2}$, $f_2(x)=\frac{1}{x}e^{-x^2/2}$. Then
	\begin{align*}
		&f_1'(x)=-\left(1-\frac{2}{(1+x^2)^2}\right)e^{-x^2/2},\\
		&f_2'(x)=-\left(1+\frac{1}{x^2}\right)e^{-x^2/2}.
	\end{align*}
	Thus
	\[
	f_1(x)=-\int_x^{\infty}f_1'(u)\md u \leq \int_x^{\infty}e^{-u^2/2}\md u\leq -\int_x^{\infty}f_2'(u)\md u=f_2(x).\qedhere
	\]
\end{proof}
\begin{lemma}\label{lowerbound}
	Let $\F$ be a bounded subset of $L^2(\Omega')$ for some probability space $\Omega'$. Then for any $\gamma>0$, $C>0$, $\inf\limits_{\xi \in \F} E [C\wedge e^{-\gamma |\xi|^2}]>0$.
\end{lemma}
\begin{proof}
	Suppose $\inf\limits_{\xi \in \F} E [C\wedge e^{-\gamma |\xi|^2}]=0$, then we can pick a sequence $\{\xi_n\}\subset \F$ such that $E [C\wedge e^{-\gamma|\xi_n|^2}]\leq 2^{-n}$. It follows for any $N\geq 1$,
	\[
	2^{-n}\geq E [C\wedge e^{-\gamma |\xi_n|^2}]\geq C\wedge e^{-\gamma N^2}P(|\xi_n|\leq N).
	\]
	We have $\sum_{n=1}^{\infty}P(|\xi_n|\leq N)<\infty$, which yields, by Borel-Cantelli lemma,
	\[
	P\left(\bigcap_{M=1}^{\infty}\bigcup_{k=N+1}^{\infty}\{|\xi_k|\leq N\}    \right)=0,
	\]
	or equivalently,
	\[
	P\left(\bigcup_{M=1}^{\infty}\bigcap_{k=N+1}^{\infty}\{|\xi_k|> N\}    \right)=1.
	\]
	This implies
	\[
	P\left(\bigcap_{N=1}^{\infty}\bigcup_{M=1}^{\infty}\bigcap_{k=N+1}^{\infty}\{|\xi_k|> N\}    \right)=1.
	\]
	That is, $|\xi_n|\to \infty$ almost surely. By Fatou's lemma,
	\[
	\infty=E\liminf_{n\to \infty}|\xi_n|\leq \liminf_{n\to \infty}E|\xi_n|\leq \sup_{\xi\in \F}\|\xi\|_{L^2},
	\]
	contradicting to the boundedness of $\F$.
\end{proof}
\begin{lemma}\label{uniformupperbounde}
	Let $\F$ be a set of $L^2$ random variables such that $E e^{\gamma \xi}<\infty$ for all $\xi\in \F$, $\gamma>0$. If $\sup\limits_{\xi\in \F}\|\xi\|_{L^2}<\infty$, then $\sup\limits_{\xi\in \F}E e^{\gamma \xi}<\infty$ holds for any $\gamma\in\mR$.
\end{lemma}
\begin{proof}
	Fix a $\xi_0\in \F$. Consider the (nonlinear) functional $F(\xi)=E e^{\gamma \xi}$. For any $\delta>0$,
	\begin{align*}
		F(\xi)-F(\xi_0)&=E[e^{\gamma \xi}-e^{\gamma \xi_0}]I_{\{\xi\neq\xi_0\}} \\
		&=E e^{\gamma \xi_0}\frac{e^{\gamma(\xi-\xi_0)}-1}{\xi-\xi_0}(\xi-\xi_0)I_{\{\xi\neq\xi_0\}} \\
		&\leq \|\xi-\xi_0\|_{L^2}\left(  E e^{2\gamma \xi_0}\left| \frac{e^{\gamma(\xi-\xi_0)}-1}{\xi-\xi_0}  \right|^2 \right)^{1/2} \\
		&\leq C\|\xi-\xi_0\|_{L^2}\left(  CE e^{2\gamma \xi_0}I_{\{|\xi-\xi_0|<\delta\}}+E e^{2\gamma \xi_0} \frac{e^{2\gamma(\xi-\xi_0)}+1}{|\xi-\xi_0|^2}I_{\{|\xi-\xi_0|\geq \delta\}}   \right)^{1/2}\\
		&\leq
		C\|\xi-\xi_0\|_{L^2}\left(  C E e^{2\gamma \xi_0}+\delta^2 E e^{2\gamma \xi_0} (e^{2\gamma(\xi-\xi_0)}+1)\right)^{1/2}\\
		&\leq C\|\xi-\xi_0\|_{L^2}\left( Ee^{2\gamma \xi_0} +\delta^2 E e^{2\gamma \xi}  \right),
	\end{align*}
	where we have used the inequality $|e^{\gamma x}-1|\leq C(\gamma)|x|$ for any sufficiently small $x$. Now letting $\delta\to 0$, taking supreme over $\F$ on left hand side gives the desired bound.
\end{proof}
Now we turn to the proof of Proposition \ref{RDUTLderivative}. In the proof, with a slight abuse of notation we still denote by $f$ the lifting of $f$. We also write $p_{\xi}$ to represent $p_{\mu}$ for any $\xi\sim \mu$. Under this convention, we have
\[
p_{\xi}(t,z)=E\left[1-\Phi\left(\frac{z-\Gamma(t)-\xi}{\sqrt{\Lambda(t)}} \right)   \right].
\]
It is straightforward to see that
\[
\partial_tp_{\xi}(t,z)=E\Phi'\left(\frac{z-\Gamma(t)-\xi}{\sqrt{\Lambda(t)}}\right)\left(\frac{z-\Gamma(t)-\xi}{2\Lambda(t)}\frac{\Lambda'(t)}{\Lambda(t)}+\frac{\Gamma'(t)}{\sqrt{\Lambda(t)}}  \right).
\]
\begin{proof}[Proof of Proposition \ref{RDUTLderivative}]
	The proof will be divided into several steps.
\vskip 5pt	
{\bf \underline{Step1:Verify (\ref{uniformcontinuity}) and (\ref{uniformboundL2})}.} We will prove that $f(t,\xi)$ is differentiable in $t$ and $\partial_tf(t,\xi)$ is uniform bounded for $t\in [t_0,t_1-\epsilon_0]$ and $\xi_0\in \cP_G(t_1)$. If this is proved, then (\ref{uniformboundL2}) is clear, and (\ref{uniformcontinuity}) is a direct consequence of mean value formula. We first prove a uniform lower bound for $p_{\xi}(t,z)$. Indeed, using Lemma \ref{estnormcdf} we have $1-N(x)\geq (1-N(x_0))\wedge e^{-|x|-\frac{x^2}{2}}\geq C\wedge C_{\delta}e^{-\frac{(1+\delta)x^2}{2}}$ for fixed $x_0$ and arbitrarily chosen $\delta>0$. Therefore,
	\begin{align*}
		p_{\xi}(t,z)\wedge (1-p_{\xi}(t,z))&\geq C_{\delta_1}Ee^{-\frac{(z-\Gamma(t)-\xi)^2}{2\Lambda(t)}(1+\delta_1)} \\
		&\geq C_{\delta_1}Ee^{-\frac{z^2}{2\Lambda(t)}(1+\delta_1)+\frac{z(\Gamma(t)+\xi)(1+\delta_1)}{\Lambda(t)}-\frac{\xi^2}{2\Lambda(t)}(1+\delta_1)} \\
		&\geq
		C_{\delta_1,\delta_2,t_0,t_1}E C\wedge e^{-\frac{z^2}{2\Lambda(t)}(1+\delta_1+\delta_2)-C_{\delta_1,\delta_2,t_0,t_1}\xi^2} \\
		&\geq
		C_{\delta_1,\delta_2,t_0,t_1}\cdot C\wedge e^{-\frac{z^2}{2\Lambda(t)}(1+\delta_1+\delta_2)},
	\end{align*}
	where Lemma \ref{lowerbound} is also used. Hence under Assumption \ref{assumptionw} we have
	\[
	w'(p_{\xi}(t,z))\leq C (C'\wedge e^{-\frac{z^2}{2\Lambda(t)}(1+\delta_1+\delta_2)})^{-(1-\beta)} \leq Ce^{\frac{z^2}{2\Lambda(t)}(1+\delta_1+\delta_2)(1-\beta)}.
	\]
	We emphasize that $\delta_1$ and $\delta_2$ are arbitrarily small constant to be determined later. Note that
	\[
	\frac{\md}{\md t}(w(p_{\xi}(t,z)))=w'(p_{\xi}(t,z))\partial_tp_{\xi}(t,z).
	\]
	Therefore,
	\begin{equation}\label{partialtest}
		\begin{aligned}
			\int_{-\infty}^{\infty}&w'(p_{\xi}(t,z))|\partial_t p_{\xi}(t,z)|u'(z)\md z\\
			&\leq CE\int_{-\infty}^{\infty}e^{\frac{z^2}{2\Lambda(t)}(1+\delta_1+\delta_2)(1-\beta)-\alpha z}\Phi'\left(\frac{z-\Gamma(t)-\xi}{\sqrt{\Lambda(t)}}\right)\left|\left(\frac{z-\Gamma(t)-\xi}{2\Lambda(t)}\frac{\Lambda'(t)}{\Lambda(t)}+\frac{\Gamma'(t)}{\sqrt{\Lambda(t)}}  \right)\right| \\
			&\leq CE\int_{-\infty}^{\infty}e^{\frac{z^2}{2\Lambda(t)}[(1+\delta_1+\delta_2)(1-\beta)+\delta_3]-\frac{(z-\Gamma(t)-\xi)^2}{2\Lambda(t)}}(C_1z+C_2\xi+C_3)\md z \\
			&\leq CE\int_{-\infty}^{\infty}e^{\frac{z^2}{2\Lambda(t)}[(1+\delta_1+\delta_2)(1-\beta)+\delta_3+\delta_4]-\frac{(z-\Gamma(t)-\xi)^2}{2\Lambda(t)}+\frac{\xi^2}{2\Lambda(t)}\delta_5}\md z \\
			&\leq CE\int_{-\infty}^{\infty}e^{-\frac{z^2}{2\Lambda(t)}(1-\tilde{\delta})+\frac{z(\Gamma(t)+\xi)}{\Lambda(t)}-\frac{(\Gamma(t)+\xi)^2}{2\Lambda(t)}+\frac{\delta_5 \xi^2}{2\Lambda(t)}}\md z\\
			&=CE\int_{-\infty}^{\infty}e^{-\frac{(1-\tilde{\delta})}{2\Lambda(t)}(z-\frac{\xi+\Gamma(t)}{1-\tilde{\delta}})^2+\frac{\frac{\tilde{\delta}}{1-\tilde{\delta}}(\Gamma(t)+\xi)^2+\delta_5 \xi^2}{2\Lambda(t)}}\md z \\
			&=C E e^{\frac{\frac{\tilde{\delta}}{1-\tilde{\delta}}(\Gamma(t)+\xi)^2+\delta_5 \xi^2}{2\Lambda(t)}}=\frac{C}{\sqrt{\mVar(\xi)}}\sqrt{\bar{\delta}}e^{C'\bar{\delta}},
		\end{aligned}
	\end{equation}
	where $\frac{1}{\bar{\delta}}=\frac{1}{\mVar(\xi)}-\frac{\tilde{\epsilon}}{\Lambda(t)}$, and $\tilde{\epsilon}=\frac{\tilde{\delta}}{1-\tilde{\delta}}+\delta_5$, $\tilde{\delta}=(1+\delta_1+\delta_2)(1-\beta)+\delta_3+\delta_4$. Thus from (\ref{partialtest}) we know
	\begin{equation}\label{partialtbound}
		\begin{aligned}
			\int_{-\infty}^{\infty}w'(p_{\xi}(t,z))&|\partial_t p_{\xi}(t,z)|u'(z)\md z\\
			&\leq \frac{C}{\Lambda(t)-\mVar(\xi)\tilde{\epsilon}}e^{\frac{C'}{\Lambda(t)-\mVar(\xi)\tilde{\epsilon}}}\\
			&\leq \frac{C}{\Lambda(t_1-\epsilon_0)-\frac{\beta\Lambda(t_1)\tilde{\epsilon}}{1-\beta}}e^{\frac{C'}{\Lambda(t_1-\epsilon_0)-\frac{\beta\Lambda(t_1)\tilde{\epsilon}}{1-\beta}}}.
		\end{aligned}
	\end{equation}
	Choosing $\delta_i$, $i=1,2,3,4,5$ all sufficiently small, then $\tilde{\epsilon}$ is sufficiently close to $\frac{1-\beta}{\beta}$. By the fact $\Lambda(t_1-\epsilon_0)>\Lambda(t_1)$ we can assume $\Lambda(t_1-\epsilon_0)-\frac{\beta\Lambda(t_1)\tilde{\epsilon}}{1-\beta}>0$. Therefore we complete the proof of Step 1 once establishing (\ref{partialtbound}).
\vskip 5pt	
{\bf \underline{Step2: Construct $f^N$ and verify (\ref{decaycondition}).}} Observe that (see (\ref{fdefRDUT}))
	\begin{equation}\label{fdefapp}
		f(t,\xi)=\int_0^{\infty}w(p_{\xi}(t,z))u'(z)\md z+\int_{-\infty}^{0}[w(p_{\xi}(t,z))-1]u'(z)\md z.
	\end{equation}
	Integrating on both sides of (\ref{growthw}) we know $w(p)\leq Cp^{\beta}$, $1-w(p)\leq C(1-p)^{\beta}$. On the other hand, when $z>0$,
	\begin{align*}
		p_{\xi}(t,z)&=P_{\eta\sim N(0,1),\xi\sim \mu}(\xi+\Gamma(t)+\sqrt{\Lambda(t)}\eta\geq z)\\
		&\leq Ce^{-(|\alpha|+1)z/\beta}(Ee^{(|\alpha|+1)\xi/\beta})(Ee^{(|\alpha|+1)\sqrt{\Lambda(t)}\eta/\beta})\\
		&\leq Ce^{-(|\alpha|+1)z/\beta},
	\end{align*}
	uniformly for $t\in [t_0,t_1-\epsilon_0]$ and $\xi\in \cP_G(t_1)$, where Lemma \ref{uniformupperbounde} is implicitly used. When $z<0$, $1-p_{\xi}(t,z)$ can be similarly estimated, with the bound $Ce^{-(|\alpha|+1)z/\beta} $ replaced by $Ce^{(|\alpha|+1)z/\beta} $. Therefore both integrands in (\ref{fdefapp}) is bounded from above by $Ce^{-|z|}$. This yields that $f$ is well-defined in $\cP_G(t_1)$. Moreover, if we define
	\begin{equation}{\label{fNdef}}
		f^N(t,\xi)=\int_0^{N}w(p_{\xi}(t,z))u'(z)\md z+\int_{-N}^{0}[w(p_{\xi}(t,z))-1]u'(z)\md z,
	\end{equation}
	(\ref{decaycondition}) is readily proved, because it has exponential decay.
\vskip 5pt	
{\bf \underline{Step3: Compute $\partial_\mu f^N$ and verify (\ref{uniformboundedness}).}}
	Once we impose the cut off given in (\ref{fNdef}), there in no singularity and it is clear that $f^N(t,\xi)$ is well-defined for any random variable $\xi$. Choose $\xi'$ such that $\|\xi'\|_{L^2}\leq 1$, let us rewrite the increments $f^N(t,\xi+\xi')-f^N(t,\xi)$ as follows:
	\begin{align*}
		f^N(t,\xi+\xi')-f^N(t,\xi)&=\int_{|z|\leq N}[w(p_{\xi+\xi'}(t,z))-w(p_{\xi}(t,z))]u'(z)\md z \\
		&=\int_{|z|\leq N}\int_0^1w'(\lambda p_{\xi+\xi'}(t,z)+(1-\lambda)p_{\xi}(t,z))\cdot(p_{\xi+\xi'}(t,z)-p_{\xi}(t,z))u'(z)\md z.
	\end{align*}
	As
	\begin{equation}\label{pincrement}
		\begin{aligned}
			|p_{\xi+\xi'}(t,z)-p_{\xi}(t,z)|&=\frac{1}{\sqrt{\Lambda(t)}}\left| E\int_0^1\Phi'\left(\frac{z-\Gamma(t)-(1-\lambda)\xi'-\xi}{\sqrt{\Lambda(t)}} \right)\cdot \xi' \md \lambda \right| \\
			&\leq C E |\xi'|\leq C\|\xi'\|_{L^2},
		\end{aligned}
	\end{equation}
	we have
	\begin{align*}
		&\left|f^N(t,\xi+\xi')-f^N(t,\xi)-\int_{|z|\leq N} w'(p_{\xi}(t,z))\cdot(p_{\xi'+\xi}(t,z)-p_{\xi}(t,z))u'(z)\md z    \right| \\
		\leq &\int_{|z|\leq N}\int_0^1 |w'(\lambda p_{\xi+\xi'}(t,z)+(1-\lambda)p_{\xi}(t,z))-w'(p_{\xi}(t,z))|\cdot|p_{\xi+\xi'}(t,z)-p_{\xi}(t,z)|u'(z)\md z \\
		\leq& C_N \|\xi\|_{L^2}\cdot \sup_{\substack{|z|\leq N \\ |x-y|\leq |p_{\xi+\xi'}(t,z)-p_{\xi}(t,z)|}}|w'(x)-w'(y)| \\
		=&o(\|\xi'\|_{L^2}).
	\end{align*}
	Using the first identity in (\ref{pincrement}) again, we get
	\[
	\begin{aligned}
		&\left|f^N(t,\xi+\xi')-f^N(t,\xi)-\frac{1}{\sqrt{\Lambda(t)}}E\left[\int_{|z|\leq N}w'(p_{\xi}(t,z))\Phi'\left(\frac{z-\Gamma(t)-\xi}{\sqrt{\Lambda(t)}} \right)u'(z)\md z\right]\cdot\xi' \right| \\
		\leq &CE\left[\int_{|z|\leq N}\int_0^1 w'(p_{\xi}(t,z))\left| \Phi'\left(\frac{z-\Gamma(t)-\xi-(1-\lambda)\xi'}{\sqrt{\Lambda(t)}} \right)-N'\left(\frac{z-\Gamma(t)-\xi}{\sqrt{\Lambda(t)}} \right)   \right|     \md z \md \lambda\right]\cdot|\xi'| \\
		\leq &C\int_{|z|\leq N}\int_0^1\int_0^1 E \left| \Phi'\left(\frac{z-\Gamma(t)-\xi-(1-\lambda)\xi'}{\sqrt{\Lambda(t)}} \right)-\Phi'\left(\frac{z-\Gamma(t)-\xi}{\sqrt{\Lambda(t)}} \right)   \right||\xi| \md \lambda'\md \lambda \md z \\
		\leq &CE \left[\int_0^1\int_0^1\Phi''\left(\frac{z-\Gamma(t)-\xi-\lambda'(1-\lambda)\xi'}{\sqrt{\Lambda(t)}}\right)|\xi'|^2 \md \lambda'\md \lambda\right] I_{\{|\xi'|\leq \|\xi'\|_{L^2}^{1/2}\}} +2CE|\xi'|I_{\{|\xi'|>\|\xi'\|^{1/2}_{L^2}\}}\\
		\leq &C((1+\|\xi'\|_{L^2}^{1/2})\|\xi'\|_{L^2}^2+\|\xi'\|_{L^2}^{3/2}) \\
		=& o(\|\xi'\|_{L^2}).
	\end{aligned}
	\]
	Recall Definition \ref{lderivativeDEF} (note also that we can directly obtain the function form of $\partial_{\mu}f^N$ stated in Proposition \ref{Lderivativefunction}), we see from the above calculation that
	\[
	\partial_{\mu}f^N(t,\mu)(y)=\frac{1}{\sqrt{\Lambda(t)}}\left[\int_{|z|\leq N}w'(p_{\mu}(t,z))\Phi'\left(\frac{z-\Gamma(t)-y}{\sqrt{\Lambda(t)}} \right)u'(z)\md z\right].
	\]
	By using Dominated convergence theorem, it is straightforward to deduce that
	\[
	\partial_v \partial_{\mu}f^N(t,\mu)(y)=\frac{1}{\sqrt{\Lambda(t)}}\left[\int_{|z|\leq N}w'(p_{\mu}(t,z))\Phi'\left(\frac{z-\Gamma(t)-y}{\sqrt{\Lambda(t)}} \right)\frac{z-\Gamma(t)-y}{\Lambda(t)}u'(z)\md z\right].
	\]
	To establish (\ref{uniformboundedness}), by tedious calculation similar to that in Step 1, we have
	\[
	|\partial_{\mu}f^N(t,\mu)(y)| \vee|\partial_v\partial_{\mu}f^N(t,\mu)(y)|\leq Ce^{\frac{\left(\frac{1-\beta}{\beta}+\hat{\epsilon} \right)y^2}{2\Lambda(t)}},
	\]
	where $\hat{\epsilon}$ is arbitrarily small. It is not hard to check that $\|\partial_{\mu} f^N(t,\mu)\|_{H^1_{\mu}}$ is uniform bounded because $\mu \in \cP_{G}(t_1)$ and the expectation with respect to $\mu$ can be computed explicitly. As for the convergence of $\partial_{\mu}f^N$, $\partial_v\partial_{\mu}f^N$ to $\partial_{\mu}f$, $\partial_v\partial_{\mu}f$ respectively, using
	\begin{align*}
		|\partial_{\mu}f^N(t,\mu)(y)-\partial_{\mu}f(t,\mu)(y)|&\leq \frac{1}{\sqrt{\Lambda(t)}}\left[\int_{|z|> N}w'(p_{\mu}(t,z))\Phi'\left(\frac{z-\Gamma(t)-y}{\sqrt{\Lambda(t)}} \right)u'(z)\md z\right] \\
		&\leq \frac{1}{N\sqrt{\Lambda(t)}}\left[\int_{-\infty}^{\infty}|z|w'(p_{\mu}(t,z))\Phi'\left(\frac{z-\Gamma(t)-y}{\sqrt{\Lambda(t)}} \right)u'(z)\md z\right],
	\end{align*}
	and $|z|\leq C_{\epsilon}e^{\epsilon z^2}$ for arbitrarily small $\epsilon>0$,  repeating the same procedure as before,  we obtain the desired convergence of $\partial_{\mu}f^N$. The convergence of $\partial_v\partial_{\mu}f^N$ is similar. Thus we have completed the proof.
\end{proof}

\section{Proof of Proposition \ref{MESprop}}\label{MESproof}
We first summarize the useful properties of $H$ in the following lemma, whose proof is quite straightforward but somewhat tedious, hence we only provide the sketch here.
\begin{lemma}\label{MESprooflm1}
	\!\!  \begin{itemize}
		\item[(1)]For any $z\neq \underline{x}$, $H(t,\cdot,z)$ is twice continuously differentiable, with $|\partial_y H(t,y,z)|\leq C/|z-\underline{x}|$, $\forall y\in \mR$.
		\item[(2)]We have the following explicit expressions:
		\begin{equation}
			\partial_y H(t,y,z)=\left\{
			\begin{aligned}
				&\frac{1}{(y-\underline{x})\sqrt{\Lambda_1(t)}}\Phi'\left(\frac{\log\left(\frac{z-\underline{x}}{y-\underline{x}}\right)-\Gamma_1(t)}{\sqrt{\Lambda_1(t)}}   \right)I_{\{y>\underline{x}\}}, &z>\underline{x}, \\
				&\frac{1}{(\underline{x}-y)\sqrt{\Lambda_1(t)}}\Phi'\left(\frac{\log\left(\frac{\underline{x}-z}{\underline{x}-y}\right)-\Gamma_1(t)}{\sqrt{\Lambda_1(t)}}   \right)I_{\{y<\underline{x}\}}, &z<\underline{x},
			\end{aligned}
			\right.
		\end{equation}
		\begin{equation}
			\partial_{yy} H(t,y,z)=\left\{
			\begin{aligned}
				&\frac{\Phi'\left(\frac{\log\left(\frac{z-\underline{x}}{y-\underline{x}}\right)-\Gamma_1(t)}{\sqrt{\Lambda_1(t)}}   \right)}{(y-\underline{x})^2\sqrt{\Lambda_1(t)}}\left[\frac{\log\left(\frac{z-\underline{x}}{y-\underline{x}}\right)-\Gamma_1(t)}{\Lambda_1(t)}-1  \right]I_{\{y>\underline{x}\}},&z>\underline{x}, \\
				\!\!\!\!\!\!\!\!&-\frac{\Phi'\left(\frac{\log\left(\frac{\underline{x}-z}{\underline{x}-y}\right)-\Gamma_1(t)}{\sqrt{\Lambda_1(t)}}   \right)}{(y-\underline{x})^2\sqrt{\Lambda_1(t)}}\left[\frac{\log\left(\frac{\underline{x}-z}{\underline{x}-y}\right)-\Gamma_1(t)}{\Lambda_1(t)}-1  \right]I_{\{y<\underline{x}\}},&z<\underline{x}.
			\end{aligned}
			\right.
		\end{equation}
		\item[(3)]For any $y\neq \underline{x}$,
		\begin{align}
			&\int_{-\infty}^{\infty}|\partial_y H(t,y,z)|\md z=\int_{-\infty}^{\infty}e^{\sqrt{\Lambda_1(t)}z+\Gamma_1(t)}\Phi'(z)\md z, \label{intpartialyH}\\
			&\int_{-\infty}^{\infty}|\partial_{yy} H(t,y,z)|\md z\leq \frac{1}{|y-\underline{x}|}\int_{-\infty}^{\infty}
			e^{\sqrt{\Lambda_1(t)}z+\Gamma_1(t)}\Phi'(z)(1+\frac{|z|}{\sqrt{\Lambda_1(t)}})\md z.\label{intpartialyyH}
		\end{align}
	\end{itemize}
\end{lemma}
\begin{proof}
	(1) and (2) are directly from the definition. To calculate the integral in (3), we first split $\mR$ into $I_-=(-\infty,\underline{x})$ and $I_+=(\underline{x},\infty)$. On $I_+$, we apply the change of variable $z'=\frac{\log(z-\underline{x})-\log(y-\underline{x})-\Gamma_1(t)}{\sqrt{\Lambda_1(t)}}$. Thus $z=\underline{x}+(y-\underline{x})e^{\sqrt{\Lambda_1(t)}z'+\Gamma_1(t)}$, $\md z=\sqrt{\Lambda(t)}(y-\underline{x})e^{\sqrt{\Lambda_1(t)}z'+\Gamma_1(t)} \md z'$. Therefore,
	\[
	\int_{I_+}|\partial_yH(t,y,z)|\md z=I_{\{y>\underline{x}\}}\int_{-\infty}^{\infty} e^{\sqrt{\Lambda_1(t)}z'+\Gamma_1(t)} \Phi'(z')\md z'.
	\]
	The integral on $I_-$ is treated similarly, while we use the change of variable $z'=\frac{\log(\underline{x}-z)-\log(\underline{x}-y)-\Gamma_1(t)}{\sqrt{\Lambda_1(t)}}$. Adding these two integrals up, we get (\ref{intpartialyH}). The proof of (\ref{intpartialyyH}) is similar.
\end{proof}
\begin{lemma}\label{measurezero}
	For $\xi\sim\mu\in \cP_{\MES}^{\delta}$, there is only one $z\in \mR$ such that $F_{X^{t,\xi}_T}(z)=\alpha_0$.
\end{lemma}
\begin{proof}
	We note that $X^{t,\xi}_T=\underline{x}+(\xi-\underline{x})e^{\eta}$, where $\eta$ is a Gaussian random variable. Thus, depending the support of $\xi$, the support of $X^{t,\xi}_T$ can be $\mR$, $(\underline{x},\infty)$, or $(-\infty,\underline{x})$. In either case, there is only one $z$ such that $F_{X^{t,\xi}_s}(z)=\alpha_0\in (0,1)$.
\end{proof}
\begin{proof}[Proof of Proposition \ref{MESprop}]
	We first write $f(t,\xi)=f_1(t,\xi)+\gamma f_2(t,\xi)$ with
	\begin{align*}
		&f_1(t,\xi)=Eh(t,\xi),\ h(t,x)=EX^{t,x}_T,\\
		&f_2(t,\xi)=\frac{1}{\alpha_0}\int_0^{\alpha_0}F^{-1}_{X^{t,\xi}_T}(\alpha) \md \alpha=\frac{1}{\alpha_0}\int_{-\infty}^{F^{-1}_{X^{t,\xi}_T}(\alpha_0)}z\md F_{X^{t,\xi}_T}(z).
	\end{align*}
	We only need to consider $f_2$ because it is clear by definition that $\partial_{\mu}f_1(t,\mu)(y)=\partial_x h(t,y)$. To deal with $f_2$, we consider the function
	$w(p)=\frac{1}{\alpha_0}(p-(1-\alpha_0))I_{\{p\geq 1-\alpha_0\}}$, and a sequence of smoothing modification $w_N\in C^{\infty}(0,1)$, such that $w(p)-w_N(p)=0$, $w'(p)-w_N'(p)=0$ with  $p$ : $|p-(1-\alpha_0)|\geq 1/N^4$ (see Figure \ref{appendixfigure} for illustration). From (\ref{FXTdef}) and (\ref{Hdef}) we know that  $X^{t,\xi}_T$ is continuous so long as $P(\xi=\underline{x})=0$, which is assured by $\xi\sim \mu\in \cP_{\MES}^{\delta}$. Thus $f_2$ can be rewritten as follows:
	\[
	f_2(t,\xi)=\int_{-\infty}^{\infty}w'(1-F_{X^{t,\xi}_T}(z))z \md F_{X^{t,\xi}_T}(z).
	\]
	\begin{figure}[!htbp]
		\centering
		\def\svgwidth{\columnwidth}
		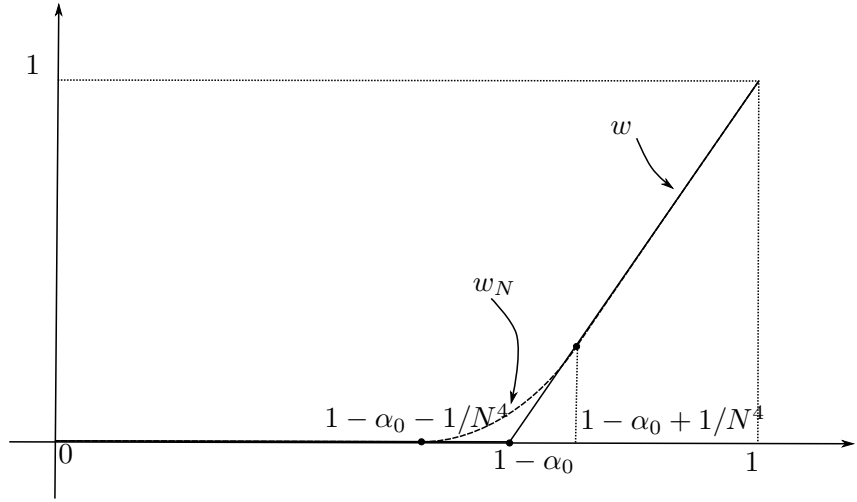
		\caption{$w$ and its smoothing $w_N$.}
		\label{appendixfigure}
	\end{figure}
	This inspires us to define
	\begin{equation}\label{MESfdef1}
		f^N_2(t,\xi)=\int_{-\infty}^{\infty}w'_N(1-F_{X^{t,\xi}_T}(z))z \md F_{X^{t,\xi}_T}(z).
	\end{equation}
	By integration by part, we further have
	\begin{equation}\label{MESfdef2}
		f^N_2(t,\xi)=\int_0^{\infty}w_N(EH(t,\xi,z))\md z+\int_{-\infty}^{0}[w_N(EH(t,\xi,z))-1]\md z.
	\end{equation}
	Using (\ref{MESfdef1}), if $\|\xi\|_{L^2}\leq C$, we have
	\begin{align*}
		|f^N(t,\xi)-f(t,\xi)|&\leq \frac{2}{\alpha_0}\int_{\{z:|F_{X^{t,\xi}_T}(z)-\alpha_0|\leq 1/N^4\}}z \md F_{X^{t,\xi}_T}(z)\\
		&\leq E X^{t,\xi}_T I_A \leq C\|\xi\|_{L^2}P(A)^{1/2} \\
		&\leq C/N^2,
	\end{align*}
	where $A=\{ |F_{X^{t,\xi}_T}(X^{t,\xi}_T)-\alpha_0|\leq 1/N^4      \}$ which has probability $2/N^4$. Thus (\ref{decaycondition}) is satisfied. Using (\ref{MESfdef2}), it is clear and straightforward to prove (\ref{uniformcontinuity}) and (\ref{uniformboundL2}) are satisfied by $f_2^N$ (on $[t,t_1]$), with constants independent of $N$. Therefore, (\ref{uniformcontinuity}) and (\ref{uniformboundL2}) are also true if we consider the limit $N\to \infty$. We now turn to prove (\ref{uniformboundedness}) and the convergence of L-derivatives. To do this, we consider any $\xi'$ with $\|\xi'\|_{L^2}\to 0$, and use the expression (\ref{MESfdef2}) to get
	\begin{align*}
		f_2^N(t,\xi+\xi')-f^N_2(t,\xi)&=\int_{-\infty}^{\infty}\left[ w_N(EH(t,\xi+\xi',z))-w_N(EH(t,\xi,z))   \right] \md z \\
		&=E\int_0^1\int_0^1\int_{-\infty}^{\infty}w'_N(\lambda EH(t,\xi+\xi',z)+(1-\lambda)EH(t,\xi,z))\\
		&\ \ \ \ \ \cdot\partial_x H(t,\xi+\lambda'\xi',z)\xi'\md z \md \lambda \md \lambda'\\
		&=E\left[ \int_{-\infty}^{\infty}w_N'(EH(t,\xi,z))\partial_yH(t,\xi,z)\md z  \right]\xi'+J_1+J_2,
	\end{align*}
	where
	\begin{align*}
		&J_1=E\int_0^1\int_0^1\int_{-\infty}^{\infty}[w'_N(\lambda EH(t,\xi+\xi',z)+(1-\lambda)EH(t,\xi,z))-w_N'(EH(t,\xi,z))]\\
		&\ \ \ \ \ \ \ \ \ \cdot \partial_y H(t,\xi+\lambda'\xi',z)\xi' \md z\md \lambda \md \lambda',\\
		&J_2=E\int_0^1\int_{-\infty}^{\infty}w_N'(EH(t,\xi,z))  [\partial_y H(t,\xi+\lambda'\xi',z)-\partial_yH(t,\xi,z)]\xi'\md z\md \lambda.
	\end{align*}
	For $J_1$, based on (1) of Lemma \ref{MESprooflm1}, we have
	\[
		|EH(t,\xi+\xi',z)-EH(t,\xi,z)|\leq CE|\xi'|/|z-\underline{x}|\leq C\|\xi'\|_{L^2}/|z-\underline{x}|,	
	\]
and $ w_N'\leq \frac{1}{\alpha_0}$.
	Thus, denoting by $I(t,z,\lambda,\lambda',\xi,\xi')$ the absolute value of the intergrands, we have
	\begin{align*}
		|J_1|&\leq\int_{-\infty}^{\infty}\int_0^1\int_0^1EI(t,z,\lambda,\lambda',\xi,\xi')\md \lambda \md \lambda' \md z \\
		&\leq \int_{\{|z-\underline{x}|\geq \|\xi'\|^{1/2}_{L^2}\}}EI(t,z,\lambda,\lambda',\xi,\xi') \md z +\int_{\{|z-\underline{x}|< \|\xi'\|^{1/2}_{L^2}\}}EI(t,z,\lambda,\lambda',\xi,\xi') \md z   \\
		&\triangleq J_3+J_4
	\end{align*}
	with
	\begin{align*}
		J_3&\leq \sup_{x,y:|x-y|\leq C\|\xi'\|_{L^2}^{1/2}}\{|w_N'(x)-w_N'(y)|\}\cdot \sup_{\lambda\in [0,1]}E\xi'\cdot\int_{-\infty}^{\infty}\partial_yH(t,\xi+\lambda\xi',z)\md z,\\
		&\leq C\|\xi'\|_{L^2}\sup_{x,y:|x-y|\leq C\|\xi'\|_{L^2}^{1/2}}\{|w_N'(x)-w_N'(y)|\} \\
		&=o(\|\xi'\|_{L^2}).
	\end{align*}
	Here (\ref{intpartialyH}) is used. On the other hand, we have
	\begin{align*}
		J_4&\leq 2\|w_N'\|_{L^{\infty}} \int_{\{|z-\underline{x}|<\|\xi'\|_{L^2}^{1/2}\}}E|\xi'|\partial_y H(t,\xi+\lambda \xi',z)\md z\\
		&\leq C\|\xi'\|_{L^2}\left(E\left[ \int_{\{|z-\underline{x}|<\|\xi'\|_{L^2}^{1/2}\}} \partial_yH(t,\xi+\lambda\xi',z)\md z \right]^2\right)^{1/2} \\
		&=o(\|\xi'\|_{L^2}),
	\end{align*}
	because we have
	\[
	E\left[ \int_{\{|z-\underline{x}|<\epsilon\}} \partial_yH(t,\xi+\lambda\xi',z)\md z \right]^2 \to 0,\
	\mbox{as} \ \epsilon \to 0
	\]
	by (\ref{intpartialyH}) and Dominated convergence theorem. For $J_2$, we define
	\[
	A(y)=\int_{-\infty}^{\infty}w'_N(EH(t,\xi,z))\partial_y H(t,y,z)\md z.
	\]
	Clearly, $A$ is continuous and bounded. We have
	\begin{align*}
		|J_2|&\leq \sup_{\lambda\in [0,1]}E|A(\xi+\lambda\xi')-A(\xi)||\xi'|\\
		&\leq \sup_{\lambda\in [0,1]}\{E|A(\xi+\lambda\xi')-A(\xi)||\xi'|I_{\{|\xi'|\leq \|\xi'\|_{L^2}^{1/2}\}}+E|A(\xi+\lambda\xi')-A(\xi)||\xi'|I_{\{|\xi'|> \|\xi'\|_{L^2}^{1/2}\}}\} \\
		&\leq \|\xi'\|_{L^2}\sup_{|x-y|\leq \|\xi'\|_{L^2}^{1/2}}+2\|A\|_{L^{\infty}}E|\xi'|I_{\{|\xi'|> \|\xi'\|_{L^2}^{1/2}\}} \\
		&\leq o(\|\xi'\|_{L^2})+O(\|\xi'\|_{L^2}^{3/2}) \\
		&=o(\|\xi'\|_{L^2}).
	\end{align*}
	As a summary, we have just proved
	\[
	\partial_{\mu}f_2^N(t,\mu)(y)=\int_{-\infty}^{\infty}w_N'(EH(t,\xi,z))\partial_yH(t,y,z)\md z.
	\]
	By (\ref{intpartialyyH}), it is not hard to prove
	\[
	\partial_v\partial_{\mu}f_2^N(t,\mu)(y)=\int_{-\infty}^{\infty}w_N'(EH(t,\xi,z))\partial_{yy}H(t,y,z)\md z,
	\]
	so long as $\xi\sim \mu\in \cP^{\delta}_{\MES}$. Using (\ref{intpartialyH}), (\ref{intpartialyyH}) and the fact that $|w'_N(p)|\leq \frac{1}{\alpha_0}$, $\forall p\in[0,1]$, we have
	\begin{align*}
		&\sup_{N\geq 1,(t,\mu)\in[t,t+\epsilon]\times \cP_{\MES}^{\delta}}E_{\xi\sim \mu}|\partial_{\mu}f_2^N(t,\mu)(\xi)|^2\leq C,\\
		&\sup_{N\geq 1,(t,\mu)\in[t,t+\epsilon]\times \cP_{\MES}^{\delta}}E_{\xi\sim \mu}|\partial_v\partial_{\mu}f_2^N(t,\mu)(\xi)|^2\leq C\delta^{-2},
	\end{align*}
	which yield (\ref{uniformbounded1}). To prove (\ref{MESpartialmuf}) and (\ref{MESpartialvpartialmuf}), we denote by $\gamma g_1(y)$ the second term of the right hand side of (\ref{MESpartialmuf}), and $\gamma g_2(y)$ the right hand side of (\ref{MESpartialvpartialmuf}), respectively. Noticing $w'(p)=\frac{1}{\alpha_0}I_{\{p\geq 1-\alpha_0\}}$, we have $g_1(y)=\int_{\infty}^{\infty}w'(EH(t,\xi,z))\partial_yH(t,y,z)\md z$, and  as $N\to \infty$,
	\begin{align*}
		E_{\xi\sim \mu}|\partial_\mu f^N_2(\xi)-g_1(\xi)|^2&\leq \int_{-\infty}^{\infty}|w'(EH(t,\xi,z))-w_N'(EH(t,\xi,z))|\partial_yH(t,\xi,z)\md z\\
		&\leq 2E\int_{\{z:|F_{X^{t,\xi}_T}(z)-\alpha_0|\leq 1/N^4\}}\partial_y H(t,\xi,z)\md z\to 0.
	\end{align*}
	Here (\ref{intpartialyH}) and Dominated convergence theorem are implicitly used together with Lemma \ref{measurezero} (so that the integration region tends to a zero measure set as $N\to 0$). Similarly,
	\[
	E_{\xi\sim \mu}|\partial_v\partial_\mu f^N_2(\xi)-g_2(\xi)|^2\leq 2E\int_{\{z:|F_{X^{t,\xi}_T}(z)-\alpha_0|\leq 1/N^4\}}\partial_{yy} H(t,\xi,z)\md z\to 0,
	\]
	by (\ref{intpartialyyH}) and $\xi\sim \mu\in\cP_{\MES}^{\delta}$.  To prove (\ref{MESpartialmufx}) and (\ref{MESpartialvpartialmufx}), we only need to apply the change of variable which has been used in the proof of Lemma \ref{MESprooflm1} and calculate the integration explicitly.
\end{proof}
\vskip 10pt
\section{Proof of Proposition \ref{propnonlinear}}
\label{proofnonlinear}
First, by Markovian property, it is clear that,
\[
f(t,\xi)=Eh_1(t,\xi)+F\left(Eh_2(t,\xi)\right),
\]
where again we adopt the convention not to distinguish between $f$ and its lifting to the functionals on $L^2$. By assumptions on $h_1$ and $h_2$ we know $\sup\limits_{\mu\in \cK}\{E|\partial_t h_i(t,\xi)|^2+E|h_i(t,\xi)|^2\}<\infty$, it is directly to show that
\[
\partial_t f(t,\xi)=E\partial_t h_1(t,\xi)+F'\left(Eh_2(t,\xi)\right )E\partial_th_2(t,\xi),
\]
and (\ref{uniformcontinuity}) and (\ref{uniformboundL2}) are satisfied. To prove other conditions needed for applying Lemma \ref{weakLdlemma}, we choose a sequence of smooth functions with compact support $\zeta^N\in C^{\infty}_c(\mR^d)$, such that $\supp(\zeta^N) \subset B(0;N+1)$, $\zeta^N\equiv 1$ in $B(0;N)$. We further assume that $\zeta^N$, $\partial_x\zeta^N$, $\partial_{xx}\zeta^N$ are all bounded, uniformly in $N$. This is possible by usual smoothing argument. We consider the cut-off
$$h_i^N(t,x)=h_i(t,x)\zeta^N(x),\ \mbox{ and}\  \ f^N(t,\xi)=Eh_1^N(t,\xi)+F\left(Eh_2^N(t,\xi)\right).$$
By the choice of $\zeta^N$, we have
\begin{eqnarray*}
	|f^N(t,\xi)-f(t,\xi)|
	&\leq &  2 E|h_1(t,\xi)|I_{\{|\xi|\geq N\}}\\
	&&+2\int_0^1|F'\left(\lambda Eh^N_2(t,\xi)+(1-\lambda)Eh_2(t,\xi)\right)|E|h_2(t,\xi)|I_{\{|\xi|\geq N\}}\md \lambda\\
	&\leq & 2\|h_1(t,\cdot)\|_{L^2_{\mu}}\frac{(E|\xi|^4)^{1/2}}{N^2}+2\sup_{|y|\leq 2\|h_2(t,\cdot)\|_{L^2_{\mu}}}|F'(y)|\|h_2(t,\cdot)\|_{L^2_{\mu}}\frac{(E|\xi|^4)^{1/2}}{N^2}\\
	&\leq & \frac{C}{N^2},
\end{eqnarray*}
where the constant $C$ is uniform in $t$ and $\xi\sim \mu\in \cK$. To proceed, we now calculate $\partial_{\mu}f^N$. For any $\xi'\in L^2$ such that $\|\xi'\|_{L^2}\leq 1$, we have
\begin{eqnarray*}
	f^N(t,\xi+\xi')&-&f^N(t,\xi)=\int_0^1E\partial_xh^N_1(t,\xi+\lambda\xi')\xi' \md \lambda \\
	&+& \int_0^1 F'\left(\lambda Eh^N_2(t,\xi+\xi')+(1-\lambda)Eh^N_2(t,\xi)\right)\cdot E\left [h^N_2(t,\xi+\xi')-h^N_2(t,\xi)\right]\md \lambda\\
	&=&J_1+J_2,	
\end{eqnarray*}
where $J_1$ and $J_2$ are treated respectively as follows:
\begin{eqnarray}\label{estJ1}
	|J_1&-&E\partial_xh^N_1(t,\xi)\xi'|\leq  \int_0^1 E|\partial_x h^N_1(t,\xi+\lambda\xi')-\partial_xh^N_1(t,\xi)||\xi'| \md \lambda \nonumber\\
	&=&\int_0^1 \left[E|\partial_x h^N_1(t,\xi+\lambda\xi')-\partial_xh^N_1(t,\xi)||\xi'|I_{\{|\xi'|\leq \|\xi'\|^{1/2}_{L^2}\}} \right.\nonumber\\
	&&\left.+E|\partial_x h^N_1(t,\xi+\lambda\xi')-\partial_xh^N_1(t,\xi)||\xi'|I_{\{|\xi'|> \|\xi'\|^{1/2}_{L^2}\}}   \right]\md \lambda\nonumber\\
	&\leq& \|\xi'\|_{L^2}\sup_{|x-y|\leq \|\xi'\|_{L^2}^{1/2}}|\partial_xh^N_1(t,x)-\partial_xh^N_1(t,y)|+2\|\partial_x h^N_1\|_{L^\infty}E|\xi'|I_{\{|\xi'|> \|\xi'\|^{1/2}_{L^2}\}}\nonumber \\
	&\leq& o(\|\xi'\|_{L^2})+C\|\xi'\|_{L^2}^{3/2}=o(\|\xi'\|_{L^2}),\\
	|J_2&-&F'(Eh^N_2(t,\xi))E\partial_xh_2^N(t,\xi)\xi'|\leq \left|\int_0^1\left[F'(\lambda Eh^N_2(t,\xi+\xi')+(1-\lambda)\right.\right.\nonumber\\
	&&\left.\left.\cdot Eh^N_2(t,\xi))-F'(Eh^N_2(t,\xi))\right]E\left[h^N_2(t,\xi+\xi')-h^N_2(t,\xi)\right]\md \lambda \right| \nonumber\\
	&&+\left| F'\left(Eh^N_2(t,\xi)\right ) E\int_0^1 \left[\partial_xh^N_2(t,\xi+\lambda \xi')-\partial_xh^N_2(t,\xi)\right]\xi' \md \lambda \right| \nonumber
    \end{eqnarray}
	\begin{eqnarray*}
	&\leq & \sup_{\substack{|x-y|\leq E|h^N_2(t,\xi+\xi')-h^N_2(t,\xi)|\\|x|,|y|\leq 2\|h^N_2\|_{L^{\infty}} }}\{|F'(x)-F'(y)|\}\cdot E|h^N_2(t,\xi+\xi')-h^N_2(t,\xi)| \\
	&+&\sup_{|x|\leq \|h^N_2\|_{L^{\infty}}}\{|F(x)|\}\int_0^1\left|\partial_xh^N_2(t,\xi+\lambda \xi')-\partial_xh^N_2(t,\xi)\right||\xi'|d\lambda \\
	&\triangleq &J_3+J_4.
\end{eqnarray*}
Using the boundedness of $\partial_x h^N_2$ and the continuity of $F'$,
\begin{eqnarray*}
	&&E|h^N_2(t,\xi+\xi')-h^N_2(t,\xi)|\leq C\|\xi'\|_{L^2},\\
	&&\sup_{\substack{|x-y|\leq E|h^N_2(t,\xi+\xi')-h^N_2(t,\xi)|\\|x|,|y|\leq 2\|h^N_2\|_{L^{\infty}} }}\{|F'(x)-F'(y)|\}=o(1), \mbox{ as $\|\xi'\|_{L^2}\to 0$     }.
\end{eqnarray*}
Therefore $J_3=o(\|\xi'\|_{L^2})$. On the other hand, using again the same event partition as in (\ref{estJ1}), we also have $J_4=o(\|\xi'\|_{L^2})$. Thus
\[
f^N(t,\xi+\xi')-f^N(t,\xi)=E\partial_xh^N_1(t,\xi)\xi'
+F'(Eh^N_2(t,\xi))E\partial_xh^N_2(t,\xi)\xi'+o(\|\xi'\|_{L^2})
\]
which, together with the definition of L-derivatives, implies
\begin{align}
	&\partial_{\mu}f^N(t,\mu)(y)=\partial_x h^N_1(t,y)+F'\left(E_{\xi\sim \mu}h^N_2(t,\xi)\right )\partial_x h^N_2(t,y),\label{NLderivativenonlinear1}\\
	&\partial_v\partial_{\mu}f^N(t,\mu)(y)=\partial_{xx} h^N_1(t,y)+F'\left(E_{\xi\sim \mu}h^N_2(t,\xi)\right )\partial_{xx} h^N_2(t,y)\label{NLderivativenonlinear2}.
\end{align}
The only thing remains to be proved is the convergence of $\partial_{\mu}f^N$ and $\partial_v\partial_{\mu}f^N$.
Indeed, as for $i=1,2$,
\begin{align*}
	&\partial_x h^N_i=\partial_x h_i\cdot \zeta^N+h_i\cdot \partial_x \zeta^N,\\
	&\partial_{xx}h^N_i=\partial_{xx}h_i\cdot\zeta^N+2\partial_{x}h_i\cdot\partial_x\zeta^N+h_i\cdot \partial_{xx}\zeta^N,
\end{align*}
we know that $h^N_i$, $\partial_x h_i^N$ and $\partial_{xx}h^N_i$ are in fact uniformly bounded, uniformly in $N$. Thus
\begin{align*}
	E|\partial_x h^N_1(t,\xi)-\partial_x h_1(t,\xi)|^2&\leq 4(E|\partial_x h_1|^2|1-\zeta^N(\xi)|^2+E|h_1(t,\xi)|^2\partial_x\zeta^N(\xi)|^2 )\\
	&\leq C\left [E|\partial_x h_1(t,\xi)|^2I_{\{|\xi|\geq N\}} +E|h_1(t,\xi)|^2I_{\{|\xi|\geq N\}}\right].
\end{align*}
Using the integrability of $h_1$,  we have $\partial_x h^N_1(t,\cdot)\to \partial_x h(t,\cdot)$ in  $L^2_{\mu}$ for any $\mu\in \cP_0$. Similarly,
\begin{align*}
	E\left|F'\left (Eh^N_2(t,\xi)\right )\partial_x h^N_2(t,\xi)\right.&\left.-F'\left(Eh_2(t,\xi)\right)\partial_x h_2(t,\xi)  \right|^2\\
	\leq& C|F'\left(Eh^N_2(t,\xi)\right)-F'\left(Eh_2(t,\xi)\right)|^2E|\partial_x h^N_2(t,\xi)|^2 \\
	&+C|F'\left(Eh_2(t,\xi)\right)|^2E|\partial_x h^N_2(t,\xi)-\partial_x h_2(t,\xi)|^2 \\
	\leq & C\sup_{\substack{|x-y|\leq E|h^N_2(t,\xi)-h_2(t,\xi)|^2\\ |x|,|y|\leq E|h(t,\xi)|^2}}\left\{|F'(x)-F'(y)|^2\right\}\cdot E|h_2(t,\xi)|^2 \\
	&+CE|\partial_x h^N_2(t,\xi)-\partial_x h_2(t,\xi)|^2\to 0,\ \mbox{as $N\to \infty$.}
\end{align*}
Letting $N$ tends to infinity on right hand side of (\ref{NLderivativenonlinear1})  proves (\ref{Lderivativenonlinear1}). Repeating the same argument to (\ref{NLderivativenonlinear2}) shows (\ref{Lderivativenonlinear2}), which completes the proof.

\end{document}